\documentclass[11pt]{article}

\usepackage[margin=1in]{geometry}
\usepackage[T1]{fontenc}

\usepackage[T1]{fontenc}
\usepackage{newpxtext}
\usepackage{mathpazo}
\usepackage{amsmath, amsthm, amsfonts, amssymb}
\usepackage{mathtools}
\usepackage{tikz}

 \usepackage[full]{complexity}
 \usepackage{hyperref}
 \usepackage{comment,bbm}
\usetikzlibrary{matrix}

\usepackage[capitalize]{cleveref}

\renewcommand{\epsilon}{\varepsilon}
\newcommand{\e}{\varepsilon}
\DeclareMathOperator*{\qp}{quasi-poly}
\DeclareMathOperator*{\polyn}{poly}
\renewcommand{\poly}{\polyn}
\newcommand{\cG}{\mathcal{G}}

\newtheorem{theorem}{Theorem}[section]
\newtheorem{lemma}[theorem]{Lemma}
\newtheorem{claim}[theorem]{Claim}

\newcommand{\calB}[0]{\mathcal{B}}

\DeclareMathOperator{\wt}{wt}
\DeclareMathOperator{\diag}{diag}

\newtheorem{corollary}[theorem]{Corollary}

\newtheorem{definition}[theorem]{Definition}

\usepackage{nicematrix}

\DeclarePairedDelimiterX{\Set}[2]{\{}{\}}{\,{#1}\,:\,{#2}\,}

\newcommand{\Z}{\mathbb{Z}}

\newcommand{\calI}[0]{\mathcal{I}}

\DeclareMathOperator*{\rank}{rank}

\title{Deterministic Constructions of Reed--Solomon Codes
Beyond the Johnson Bound}

\author{Abhranil Chatterjee\thanks{Indian Institute of Technology Kharagpur. Email: \texttt{abhranilc@cse.iitkgp.ac.in}} 
\and
Sumanta Ghosh\thanks{Indian Statistical Institute, Kolkata. Email: \texttt{besusumanta@gmail.com}} 
\and
Zeyu Guo\thanks{The Ohio State University. Email: \texttt{zguotcs@gmail.com}}
\and
Rohit Gurjar\thanks{Indian Institute of Technology Bombay. Email: \texttt{rgurjar@iitb.ac.in}.} 
\and
Roshan Raj\thanks{The Ohio State University. Email: \texttt{raj.144@osu.edu}.} 
\and 
Zihan Zhang \thanks{Simons Institute for the Theory of Computing, Berkeley and Institute for Advanced Study, Princeton. Email: \texttt{ zzhsdj@foxmail.com}.}}
\date{}

\begin{document}

\maketitle  
\begin{abstract}
We give a deterministic quasipolynomial-time construction of
Reed--Solomon codes that are list-decodable beyond the Johnson
radius.
Specifically, for every $\varepsilon\in(0,1]$ and all sufficiently
large $n$, we construct an $[n,k]$ Reed--Solomon code of rate
$R=\Omega(\varepsilon/(1+\log(1/\varepsilon)))$ over a field of size
$(nk)^{\operatorname{quasi\text{-}poly}(nk)}$ that is
$(1-\varepsilon,O(1/\varepsilon))$ list-decodable.
The construction runs in
$\operatorname{quasi\text{-}poly}(nk)$ time.

We also give a deterministic construction of $[n,k]$-MDS$(3)$
Reed--Solomon codes for all integers $1\leq k<n$, with the same
asymptotic bounds on field size and construction time.
In the constant-rate regime, our field-size bound substantially
improves the bound of $O\left(n^{k^{2k}}\right)$ achieved by Roth's construction (Roth, IEEE TIT 2022).

Our constructions exploit a connection to polynomial identity
testing for read-once determinants, together with the
quasipolynomial isolating-weight construction of Gurjar and
Thierauf for linear matroid
intersection (Gurjar and Thierauf, STOC 2017 and Computational Complexity 2020).
\end{abstract}
\newpage
\section{Introduction} 

Reed--Solomon (RS) codes are among the most fundamental and widely studied families of error-correcting codes.
A message is represented by a low-degree polynomial over a finite field $\mathbb F_q$, and redundancy is introduced by evaluating it at more points than are needed to determine it.
Formally, for distinct elements $\alpha_1,\dots,\alpha_n\in\mathbb F_q$, called the evaluation points, the corresponding $[n,k]$ RS code is
\[
 \mathrm{RS}_{n,k}(\alpha_1,\dots,\alpha_n)
 :=\{(f(\alpha_1),\dots,f(\alpha_n)):
 f\in\mathbb F_q[X],\ \deg f<k \}.
\]
It has rate $R=k/n$ and minimum distance $d=n-k+1$, attaining the \emph{Singleton bound} $d\leq n-k+1$ for linear codes \cite{Singleton64}.
Thus, RS codes achieve the optimal tradeoff between rate and minimum distance, and hence between rate and unique-decoding radius $\lfloor(d-1)/2\rfloor$.
This optimality, together with their simple algebraic structure, has led to numerous applications in both theory and practice.

\paragraph*{List decoding and the generalized Singleton bound} 
One limitation of unique decoding is that its decoding radius must be strictly less than half the minimum distance.
List decoding, introduced independently by Elias~\cite{Elias57} and Wozencraft~\cite{Wozencraft58}, relaxes this model by allowing the decoder to return a short list of candidate messages that includes the transmitted message.
This relaxation makes it possible to tolerate substantially more errors.
Beyond their role in error correction, list decoding and the related notion of list-recovery have become important tools in complexity theory, pseudorandomness, and cryptography.
Applications include worst-case to average-case reductions~\cite{CPS99}, hardness amplification~\cite{STV01,IJK09}, and explicit constructions of pseudorandom generators and hitting set generators~\cite{SU05,Uma03,Uma09} as well as randomness extractors~\cite{Tre01,TZ04,TZS06,SU05,GUV09}.
Cryptographic applications include hard-core bits~\cite{GL89,AGS03} and traitor-tracing schemes~\cite{SSW03}.

Formally, a code $C\subseteq\mathbb F_q^n$ is said to be
\emph{(combinatorially) $(\rho,L)$ list-decodable} if every Hamming ball
of radius $\rho n$ in $\mathbb F_q^n$ contains at most $L$ codewords of
$C$. Equivalently, for every $y\in\mathbb F_q^n$ and every choice of
$L+1$ distinct codewords $c_1,\dots,c_{L+1}\in C$,
\[
\max_{1\leq i\leq L+1} d_H(c_i,y)>\rho n,
\]
where $d_H$ denotes the Hamming distance.
The stronger notion of \emph{average-radius list-decodability}
replaces the maximum distance with the average distance:
$C$ is \emph{average-radius $(\rho,L)$ list-decodable} if
\[
\frac{1}{L+1}\sum_{i=1}^{L+1} d_H(c_i,y)>\rho n
\]
for every $y\in\mathbb F_q^n$ and every choice of $L+1$ distinct
codewords $c_1,\dots,c_{L+1}\in C$.
In both definitions, $\rho$ is called the (relative) list-decoding radius,
and $L$ is called the list size.

One natural question is to determine the best possible tradeoff
between the list-decoding radius and the rate.
Assuming that $\rho n$ is an integer, the
\emph{generalized Singleton bound}~\cite{ST23,Rot22} states that
an $[n,k]$ linear code over $\mathbb F_q$ of rate $R=k/n$ can be
$(\rho,L)$ list-decodable\footnote{Here we assume $q>L$ and $L\leq k<n$;
see~\cite[Theorem~3 and Corollary~5(a)]{Rot22}.} only if
\[
  \rho \leq \frac{L}{L+1}(1-R).
\]
The same upper bound applies to average-radius list decoding,
since it is a stronger requirement.
When $L=1$, it recovers the unique-decoding benchmark $(1-R)/2$
implied by the classical Singleton bound.

\paragraph*{The Johnson bound}
The generalized Singleton bound gives an upper bound on the
list-decoding radius. On the other hand, the classical
\emph{Johnson bound}~\cite{Joh62,GS01} gives a lower bound that
depends only on the minimum distance.
Specifically, a code of relative minimum distance $\delta\in[0,1]$
is list-decodable at any radius $\rho<J(\delta)$, where $J(\delta):=1-\sqrt{1-\delta}$,
with list size bounded solely in terms of the gap $J(\delta)-\rho$.
Since an $[n,k]$ RS code has relative minimum distance $(n-k+1)/n$,
its Johnson radius is $1-\sqrt{(k-1)/n}$, or approximately
$1-\sqrt R$. In comparison, the generalized Singleton bound
approaches $1-R$ as the list size grows.

Whether there exist RS codes that are list-decodable beyond the
Johnson radius was a longstanding open problem, even for randomly
chosen evaluation points and even at the combinatorial level.
Shangguan and Tamo~\cite{ST23} conjectured that random RS codes
over sufficiently large fields not only beat the Johnson bound
but attain the generalized Singleton bound for every fixed list
size $L$.
A long line of work~\cite{RW14,ST23,DBLP:journals/siamcomp/GuoLSTW24,FKS22,GST23}
established increasingly strong list-decoding guarantees beyond
the Johnson radius.
Brakensiek, Gopi, and Makam~\cite{BGM24} ultimately resolved
the conjecture, showing that random RS codes over sufficiently
large fields attain the generalized Singleton bound with high
probability, even in the stronger average-radius sense.
In particular, random RS codes achieve \emph{list-decoding capacity},
tolerating a $1-R-\varepsilon$ fraction of errors with list size
$O(1/\varepsilon)$.

The result of Brakensiek, Gopi, and Makam uses fields of exponential
size in the code length.
Guo and Zhang~\cite{GZ23} subsequently showed that fields of size
$O_{\varepsilon}(n^2)$ suffice for achieving list-decoding capacity
with list size $O(1/\varepsilon)$.
Alrabiah, Guruswami, and Li~\cite{AGL24} further reduced the field
size to $O_{\varepsilon}(n)$ while preserving this list-size bound.
Thus, the combinatorial list-decodability of RS codes with random
evaluation points is now well understood.

\paragraph*{Deterministic constructions}
Random RS codes can surpass the Johnson bound and even attain
the generalized Singleton bound. This motivates the search for
efficient deterministic constructions of evaluation points
$\alpha_1,\ldots,\alpha_n$ with comparable guarantees, including
in small characteristic.\footnote{
We use the term ``deterministic construction'' rather than
``explicit construction,'' which is often understood to require
polynomial running time. The running times of our constructions
are stated explicitly.
}
Beyond its coding-theoretic significance, this problem provides
a concrete setting for developing derandomization techniques
for structured algebraic problems.

This construction problem is closely connected to several other
questions in coding theory.
Brakensiek, Gopi, and Makam~\cite{BGM22,BGM24} established
equivalences, up to duality, between optimal average-radius
list-decoding guarantees and several other properties.
One is the higher-order MDS property, denoted
$\operatorname{MDS}(\ell)$: intersections of $\ell$ column spans
of a generator matrix must have the same dimensions as the
corresponding intersections for a generic matrix.
Another asks that a fixed code realize every admissible zero
pattern of the corresponding order through a suitable choice
of generator matrix, connecting this question to the GM-MDS
literature~\cite{DSY14,Lov18,YH19}.
The same framework also characterizes maximal recoverability
of tensor codes with a single parity check per column:
such codes must recover from every erasure pattern that can
be corrected by some tensor code with the same row and column
parameters~\cite{GHKSSY17,BGM22}.
The equivalence among these properties allows deterministic constructions to be
transferred between the different settings, so progress on one
construction problem can yield progress on the others.

Shangguan and Tamo~\cite{ST23} gave a deterministic construction of RS codes
attaining the generalized Singleton bound for list sizes two and
three, although their field sizes are doubly exponential in the
constant-rate regime.
Roth~\cite{Rot22} reduced the field-size requirement for list size
two to polynomial in $n$ when the code dimension $k$ is fixed.
Combined with the main result of Brakensiek, Gopi, and
Makam~\cite{BGM24}, the field-extension method of Shangguan and
Tamo also yields deterministic constructions of
$[n,k]$-$\operatorname{MDS}(\ell)$ RS codes for arbitrary order
$\ell$, over fields of size $2^{k^n}$.
Via duality, these provide optimal average-radius list-decoding
guarantees for all list sizes up to $\ell-1$.
Building on the methods of Shangguan--Tamo and Roth,
Brakensiek, Dhar, and Gopi~\cite{BDG23} obtained constructions
over fields of size $n^{(\ell k)^{O(\ell k)}}$.
This gives polynomial-size fields when $k$ and $\ell$ are fixed,
but the bound remains doubly exponential in the constant-rate
regime.

In a recent breakthrough, Brakensiek, Chen, Putterman,
Zhang, and Zheng~\cite{BCPZZ26} showed that RS codes over sufficiently
large prime fields can be efficiently list decoded up to capacity
for every choice of distinct evaluation points, with polynomial
list size.
This removes the need to choose special evaluation points in
that regime, but relies on large characteristic and does not
provide the optimal constant-list-size guarantees available
for random RS codes.
This motivates the problem of obtaining deterministic constructions
with small-list guarantees in arbitrary characteristic while
keeping both the field size and the construction time small.
\subsection{Main Results}

Our main contribution is a deterministic quasipolynomial-time
construction of evaluation points that simultaneously satisfy a
large family of determinantal nonvanishing conditions arising
in the study of list decoding of RS codes.
This yields RS codes that are list-decodable beyond the Johnson
radius. 
More precisely, we prove the following theorem.

\begin{theorem}\label{thm:beyondJohnson}
For every $\e\in(0,1]$ and all sufficiently large integers $n$,
there is a deterministic algorithm running in $\qp(nk)$ time
that constructs an $[n,k]$ RS code of rate $=\frac{k}{n}=\Omega( \frac{\e}{\log(1/\e)+1})$ over a field of size at most $(nk)^{\qp(nk)}$ that is
$(1-\e,O(1/\e))$ list-decodable.
\end{theorem}
In particular, for sufficiently small $\varepsilon>0$ and all
sufficiently large $n$, the construction satisfies
$(k-1)/n>\varepsilon^2$. Hence its list-decoding radius
$1-\varepsilon$ exceeds the Johnson radius
$1-\sqrt{(k-1)/n}$.

Our second result concerns higher-order MDS codes.
We first recall their formal definition.
\begin{definition}[MDS$(\ell)$ code~\cite{BGM24}]
Let $\ell\geq 1$ be an integer, and let $C$ be an $[n,k]$ linear
code over $\mathbb F$ with generator matrix
$G\in\mathbb F^{k\times n}$.
Let $W$ be a $k\times n$ matrix whose $(i,j)$-th entry is a variable $x_{i,j}$.
For $A\subseteq[n]$, let $G_A$ and $W_A$ denote the subspaces
spanned by the columns indexed by $A$, over $\mathbb F$ and
over the function field $\mathbb F(\{x_{i,j}\}_{i\in [k],j\in [n]})$, respectively.
We say that $C$ is an \emph{MDS$(\ell)$ code} if
\[
  \dim\left(\bigcap_{i=1}^{\ell}G_{A_i}\right)
  =
  \dim\left(\bigcap_{i=1}^{\ell}W_{A_i}\right)
\]
for every choice of subsets $A_1,\dots,A_\ell\subseteq[n]$.
\end{definition}

The MDS$(1)$ and MDS$(2)$ properties are both equivalent to
the ordinary MDS property, so MDS$(3)$ is the first nontrivial
higher-order case.
Even in this case, a substantial gap remains between the known
lower bounds on the required field size and the field sizes
achieved by deterministic constructions.
Our second result gives a deterministic quasipolynomial-time
construction of MDS$(3)$ RS codes and, in the constant-rate
regime, improves the field-size bound of
$O\left(n^{k^{2k}}\right)$ achieved by Roth~\cite{Rot22}.

\begin{theorem}\label{thm:MDS3}
For all integers $1\leq k<n$, there is a deterministic algorithm
running in $\qp(nk)$ time that constructs an $[n,k]$ RS code
that is MDS$(3)$ over a field of size at most $(nk)^{\qp(nk)}$.
\end{theorem}

\paragraph{Independent and concurrent work.}
Our results were obtained independently of concurrent work by
Alrabiah, Goyal, and Guruswami~\cite{AGG26}.\footnote{A version
of this paper containing our results was submitted to STACS~2027
before~\cite{AGG26} appeared on ECCC.}
They give deterministic constructions of Reed--Solomon codes
that achieve list-decoding capacity with optimal list size.
For every fixed gap $\eta>0$, their construction gives codes
of rate $R$ that are
$(1-R-\eta,O((1-R)/\eta))$ list-decodable, with field size
and construction time $n^{O_\eta(1)}$. They also give
deterministic list-decoding algorithms. In particular, for
every fixed $\varepsilon>0$, their results give
$(1-\varepsilon,O(1/\varepsilon))$ list-decodable RS codes
of rate $\Theta(\varepsilon)$, improving the rate, field size,
and construction time in \cref{thm:beyondJohnson}.

Our approach is different.  We use isolating weights for
matroid intersection to construct a Vandermonde substitution
that preserves nonzeroness for a class of read-once determinant
polynomials (\cref{thm:RSreadonce}). This also gives the exact MDS(3)
construction in \cref{thm:MDS3}, which is not established by the
fixed-gap list-decoding guarantees of~\cite{AGG26}.

\subsection{Proof Idea}
Both of our main results follow from a common polynomial identity
testing problem. The known determinantal characterizations of
list-decodable and higher-order MDS RS codes \cite{BGM22, BGM24} reduce the
desired coding-theoretic properties to preserving the nonzeroness of
certain determinant polynomials. After elementary row or column
operations, these polynomials take the form
\[
\det(U_1X_JU_2^{\mathsf T}),
\]
where $X_J$ is diagonal and the first matrix $U_1$ represents a
partition matroid  defined below (or, see \cref{def:partitionmatroid}). Since each variable occurs only once in $X_J$,
these are called read-once determinant polynomials.
\begin{definition}
    Let $k,n$ be positive integers and $J$ be a subset of $[n]$. We define $\mathcal{M}_{k,J}$ to be the partition matroid of rank $|J|$ on the ground set $[k]\times J$, whose $|J|$ parts are
$Q_j:=\{(i,j)\mid i\in[k]\}$, where $j\in J$,
with each part having rank one. 
\end{definition}

Precisely, we show the following theorem, which we later use for our deterministic construction of RS codes.
\begin{theorem}\label{thm:RSreadonce}
Let $k,n$ be positive integers with $k\leq n$, and let $X$ be the
$kn\times kn$ diagonal matrix whose rows and columns are indexed by
$[k]\times[n]$, with $X[(i,j),(i,j)]=x_{i,j}$ for every $(i,j)\in[k]\times[n]$.
For a subset $J$ of $[n]$, let $X_J$ denote the  principal submatrix of $X$ with rows and columns indexed by $[k]\times J$.

Then there exist positive integers $C,D\leq (nk)^{O(\log(nk))}$ such that the
following holds. For every prime power $q\geq C$ and every positive
integer $r\geq D$, there exist pairwise distinct
$
\alpha_1,\ldots,\alpha_n\in\mathbb{F}_{q^r}
$ such that, for every subset $J$ of $[n]$ and every  pair of $|J|\times k|J|$ matrices
$U_1,U_2$ over $\mathbb{F}_q$
whose columns are indexed by $[k]\times J$ with $U_1$ representing the matroid $\mathcal{M}_{k,J}$, if 
$\det(U_1X_JU_2^{\mathsf T})$ is not the zero polynomial, then it evaluates to a nonzero element of
$\mathbb{F}_{q^r}$ under the substitution
\[
x_{i,j}=\alpha_j^{\,i-1}
\text{ for every }(i,j)\in[k]\times J.
\]
Moreover, such $\alpha_1,\alpha_2,\dots, \alpha_n$ can be constructed in $\qp(nk)\cdot \poly(qr)$ time.
\end{theorem}
To prove \cref{thm:RSreadonce}, we first use the Cauchy--Binet formula for the determinant of product of two rectangular matrices to
identify the monomials in the support of
$\det(U_1X_JU_2^{\mathsf T})$ with the common bases of two linear
matroids represented by the columns of $U_1$ and $U_2$, respectively. Thus, for an isolating weight vector $w\in \Z^{[k]\times J}$ for the corresponding matroid-intersection instance (see \cref{def:isolating}), the substitution $x_{i,j}\mapsto t^{w_{i,j}}$ prevents cancellation among the monomials, and preserves the nonzeroness of $\det(U_1X_JU_2^{\mathsf T})$.

The work of Gurjar and Thierauf~\cite{DBLP:journals/cc/GurjarT20} provides a construction of isolating weight assignment for common bases of linear matroid intersections. However, we cannot use their construction directly. The issue is that it does not immediately have the special form
required by the Vandermonde substitution
$x_{i,j}=\alpha_j^{i-1}$. For this, the weights must have the form \(w_{i,j}=(i-1)\widetilde w_j\), so that the substitution \(x_{i,j}=z^{w_{i,j}}\) can be realized by setting the RS evaluation point \(\alpha_j=z^{\widetilde w_j}\). In general, for any two arbitrary linear matroids, it is not clear to us how to construct isolating weight assignment from ~\cite{DBLP:journals/cc/GurjarT20} that is compatible with Vandermonde substitution. However, for our purposes, we observe that among the two matroids, one  is the special partition matroid defined above. This helps us to construct an isolating weight assignment of our desired form. 

One of our main technical steps is to construct a quasipolynomial-size family $\mathcal W$ of weight vectors whose weights are quasi-polynomially large and have the form $w_{i,j}=(i-1)\widetilde w_j$
with the following property: For any two matroids having a common base, one of which is $\mathcal M_{k,J}$, the family $\mathcal W$ contains an isolating
weight vector for their common bases.
We then combine all weights in this family via polynomial interpolation and obtain a bivariate
 map. Then, we evaluate this bivariate polynomial map over a sufficiently large extension field.  

Finally, in \cref{sec:RSLD,sec:RSMDS3}, we show that the determinant
polynomials appearing in the known characterizations of list
decoding (\cref{lem:intersection-matrix-factorization}) and MDS$(3)$ codes (\cref{lem:3MDS-matrix-factorization}) belong to the class covered by
\cref{thm:RSreadonce}.

\section{Preliminaries and Notation}
For a matrix  $U$ with indexed rows and columns, and for subsets $S$ of the row indices and $T$ of the column indices, we denote by $U[S,T]$  the submatrix of $U$ consisting of the rows indexed by $S$ and the columns indexed by $T$. $U[T]$ denotes the submatrix $U[R,T]$ where $R$ is  the entire row-index set.

For matrices $A\in\mathbb{F}^{m_1\times n_1}$ and
$B\in\mathbb{F}^{m_2\times n_2}$, we denote by $A\oplus B$ their
direct sum, defined as
\[
A\oplus B:=
\begin{bmatrix}
A & 0\\
0 & B
\end{bmatrix}
\in\mathbb{F}^{(m_1+m_2)\times(n_1+n_2)}.
\]
When the rows and columns are indexed, we list the row and column
indices of $A$ before those of $B$.
\paragraph*{Matroids} A \emph{matroid} $\mathcal M=(E,\mathcal I)$ consists of a finite ground set $E$ and a nonempty family $\mathcal I\subseteq 2^E$ satisfying the following two properties:
\begin{enumerate}
    \item \textit{Closure under subsets:} If $S\in \mathcal{I}$ and $T\subseteq S$, then $T\in \mathcal{I}$.
    \item \textit{Augmentation Property:} If $|S|>|T|$ and $S,T\in \mathcal{I}$, then there exists $e\in S\setminus T$ such that $T\cup \{e\}\in \mathcal{I}$.
\end{enumerate}
The set family $\calI$ is called the \emph{independent set family} of $\mathcal{M}$. The augmentation property ensures that all the maximal independent sets of $\mathcal{M}$ have the same size. The collection $\calB$ of all the maximal independent sets is called the \emph{base family} of $\mathcal{M}$. The base family $\calB$ satisfies the following property:  
\begin{center}
\textit{Base exchange property:} Let $B,B'\in\calB$. Then for all $a\in B\setminus B'$ there exists a $b\in B'\setminus B$ such that $B-a+b$ is in $\calB$. 
\end{center}
Given the base family $\calB$ of a matroid $\mathcal{M}$, its independent set family is determined as $\calI=\{I\,\mid\,\exists B\in\calB,\  I\subseteq B\}$. Therefore, a matroid $\mathcal{M}$ can be represented as $(E,\calI)$ or $(E,\calB)$. In this work, we mostly use $\mathcal{M}=(E,\calB)$ to represent a matroid. Every matroid $\mathcal{M}$ is associated with a function, $\rank:2^E\to \mathbb{N}$, defined as
\[
\rank(S)= \max \{|T|\,\mid\, T\subseteq S,\ T\in \mathcal{I}\}.
\]
 The rank of the ground set $E$ is called the rank of the matroid $\mathcal{M}$. It is equal to the cardinality of the bases. For more details on matroids, one can see some excellent textbooks like \cite{Oxl06, Sch03B}.
\paragraph*{Linear matroids} A well-known class of matroids is that of \emph{linear matroids}. A linear matroid over a field $\mathbb{F}$ is represented by an $r\times m$ matrix $U$ over the field $\mathbb{F}$ with the full row rank. Assume that the columns are indexed by $[m]$, which is the ground set of the matroid. Let $\calB=\{B\subseteq [m]\,\mid\,|B|=r,\ \det(U_B)\neq 0\}$. It is not hard to prove that $\mathcal{M}=([m],\calB)$ is a matroid with $\calB$ as the base family. We say that $U$ \emph{represents} the matroid $\mathcal{M}$ over field $\mathbb{F}$.

\paragraph*{Partition matroids}\label{def:partitionmatroid} Another class of matroids crucial for our results is that of
\emph{partition matroids}. A matroid $\mathcal M$ on a ground set
$E$ is a partition matroid if there exist a partition
$Q_1\sqcup\cdots\sqcup Q_n=E$ and nonnegative integers
$h_1,\ldots,h_n$, with $h_j\leq |Q_j|$ for every $j\in[n]$,
such that the bases of $\mathcal M$ are precisely the subsets
$B\subseteq E$ satisfying $|B\cap Q_j|=h_j$
for every $j\in[n]$.
In particular, the rank of each part $Q_j$ is $h_j$.

\paragraph*{Matroid intersection} Let $\mathcal{M}_1=(E,\calB_1)$ and $\mathcal{M}_2=(E,\calB_2)$ be two matroids defined on the same ground set $E$. The problem of finding a common base is called the \emph{matroid intersection} problem. The problem of perfect matching for bipartite graphs and many other problems can be viewed as special cases of the matroid intersection problem.

\paragraph*{Matroid intersection polytope} Let $\mathcal{M}_1$ and $\mathcal{M}_2$ be two matroids of rank $r$ on a common ground set $E$ with rank function $r_1$ and $r_2$, respectively. Then, the convex hull of characteristic vectors of common bases in $\mathbb{R}^{|E|}$, denoted by $P(\mathcal{M}_1,\mathcal{M}_2)$, is described by the following linear program \cite{Sch03B}. 
\[x_e\geq 0 \text{ for each } e\in E, x(S)=\sum_{e\in S}x_e\leq r_i(S) \text{ for each } S\subseteq E, i\in \{1,2\}, \text{ and } x(E)=r. \]

\subsection{Isolating Weight Assignment for Matroid Intersection}

Throughout this paper, faces of a polytope are assumed to be nonempty,
and the entire polytope is allowed as a face.

For a weight vector $w\in \Z^E$, the \emph{weight} of a subset $S\subseteq E$, denoted by $w(S)$, is $\sum_{e\in S}w_e$.

\begin{definition}\label{def:isolating} Let $\mathcal{M}_1$ and $\mathcal{M}_2$ be two matroids defined on a common ground set $E$ and  $w\in \mathbb{Z}^E$ be a weight vector. $w$ is called \emph{isolating} for $\mathcal{M}_1$ and $\mathcal{M}_2$ if there exists a unique maximum weight common base in $\mathcal{M}_1$ and $\mathcal{M}_2$ with respect to $w$. A family of weight vectors $\mathcal{W}\subseteq \Z^{E} $ is called \emph{isolating} for $\mathcal{M}_1,\mathcal{M}_2$ if there exists $w\in \mathcal{W}$ that is isolating for $\mathcal{M}_1,\mathcal{M}_2$. \end{definition}

The following lemma \cite[Lemma 3.5]{DBLP:journals/cc/GurjarT20} gives a characterization for the tight inequalities in the LP for faces of the matroid intersection polytope.
\begin{lemma}\label{lem:FaceMI}
Let $\mathcal{M}_1$ and $\mathcal{M}_2$ be two matroids defined on a common ground set $E$ that have a common base and have rank functions $r_1$ and $r_2$, respectively. Let $F$ be a face of $P(\mathcal{M}_1,\mathcal{M}_2)$. Then, there exist two partitions $\mathcal{P}_1,\mathcal{P}_2$ of $E$ such that for any $S_1\in \mathcal{P}_1$ and $S_2\in \mathcal{P}_2$, there exists $\nu_{S_1},\mu_{S_2}\in \mathbb{N}$ such that \[x(S_1)=\nu_{S_1} \text{ and } x(S_2)=\mu_{S_2} \text{ for every } x\in F.\]
Moreover,\begin{enumerate}
    \item If for some $T\subseteq E$, $x(T)=r_1(T)$ for all $x\in F$, or $x(T)=r_2(T)$ for all $x\in F$, then $T$ is the union of sets from $\mathcal{P}_1$, respectively $\mathcal{P}_2$,
    \item If for some $e\in E, x_e=0 $ for all $x\in F$, then  there is a $S_1\in \mathcal{P}_1$ and $S_2\in \mathcal{P}_2$ such that $S_1=S_2=\{e\}$ and $\nu_{S_1}=\mu_{S_2}=0$.
\end{enumerate}
\end{lemma}
\begin{definition}\label{def:cycle}
   Let $\mathcal{M}_1$ and $\mathcal{M}_2$ be two matroids defined on a common ground set $E$ with a common base and  $F$ be a face of the polytope. Let $\mathcal{P}_1,\mathcal{P}_2$ be the partitions of $E$ from \cref{lem:FaceMI}.  A sequence
$C = (e_1,e_2,\dots,e_{2p})$ of distinct elements of $E$, where $p\geq 1$ and $e_{2p+1}:=e_1$, is called a \emph{cycle}
with respect to face $F$, if consecutive pairs are alternately in a set
from $\mathcal{P}_1$ and a set from $\mathcal{P}_2$. That is, for $i \in [p]$,
\begin{align*}
      e_{2i-1},e_{2i} &\in S_{2i-1} \text{ for some } S_{2i-1} \in \mathcal{P}_1 \\
      e_{2i},e_{2i+1} &\in S_{2i} \text{ for some } S_{2i}\in \mathcal{P}_2.
\end{align*}
For a face $F$, we denote by $\mathcal{C}_F$ the set of all cycles with respect to the face $F$. In particular,
$\mathcal C_{P(\mathcal M_1,\mathcal M_2)}$ is the set of
cycles associated with the entire polytope $P(\mathcal M_1,\mathcal M_2)$.

\end{definition}
We say a cycle is \emph{simple} if each part of the two partitions contains at most two elements from the cycle.
For a weight vector $w\in \Z^{E}$ and a cycle $C=(e_1,\dots,e_{2p})$, we define its \emph{circulation}, denoted by $c_w(C)$, as \[c_w(C)=|w(e_1)-w(e_2)+w(e_3)-\cdots-w(e_{2p})|\] 
Now, the following lemma gives the construction of an isolating set of weight vectors which follows from \cite[Section 3.4, Lemma~3.19]{DBLP:journals/cc/GurjarT20}.
\begin{lemma} \label{lem:isolationset}
    Let $\mathcal{M}_1,\mathcal{M}_2$ be two matroids with common ground set $E$ that have a common base. Let $\mathcal{U}\subseteq \Z^{E}$ such that for any face $F$ of $P(\mathcal{M}_1,\mathcal{M}_2)$ and any collection $\mathcal C$ of at most $|E|^4$ simple cycles of $\mathcal{C}_F$, there exist $u\in \mathcal{U}$ such that $c_u(C)\neq 0 \text{ for each cycle } C\in \mathcal{C}.$ 
    Let $\ell$ and $N$ be positive integers such that \[\ell\geq \lceil \log |E|\rceil-1\text{ and } N>|E|\cdot\max\limits_{u\in \mathcal U, e\in E} | u_e|.\] Let $\mathcal{W}\subseteq \Z^E$ be a set of weight functions defined as \begin{equation}\label{eq:defWkn}
    \mathcal{W} = \{ N^{\ell} u^{0} + N^{\ell-1}u^1 +\dots+ N^0u^{\ell} \mid u^i\in \mathcal{U} \text{ for each } i \in \{0,1,\dots,\ell\}\}.
    \end{equation} Then, $\mathcal{W}$ is isolating for $\mathcal{M}_1,\mathcal{M}_2$.
\end{lemma}


\section{Proof of \texorpdfstring{\cref{thm:RSreadonce}}{Theorem~5}} \label{sec:RSreadonce}
We first prove the theorem for $J=[n]$.  Fix matrices $U_1,U_2\in\mathbb{F}_q^{n\times kn}$ satisfying the
hypotheses of the theorem, and suppose that
$\det(U_1XU_2^{\mathsf T})$ is not the zero polynomial. We start by studying the monomials of $\det({U_1}XU_2^{\mathsf T})$. By the Cauchy--Binet formula for the determinant of the product of two rectangular matrices, we get
\[
\det({U_1}XU_2^{\mathsf T})=\sum_{S\subseteq [k]\times [n], |S|=n} \det({U_1}[S])\det(U_2[S])\prod_{(i,j)\in S}x_{i,j}.
\]
This implies that $\det({U_1}XU_2^{\mathsf T})$ is a multilinear polynomial of degree $n$. Let  \begin{equation} \label{eq:monsupdef}
    \mathcal{B}=\{S\subseteq [k]\times [n]\mid |S|=n \text{ and } \det({U_1}[S])\det(U_2[S])\neq 0 \}.
\end{equation}
In other words, $\mathcal{B}$ is the family of subsets of $[k]\times [n]$ consisting of support of monomials of $\det({U_1}XU_2^{\mathsf T})$. 
Let $\mathcal{M}_{U_2}$ denote the matroid represented by $U_2$ on the ground set $[k]\times [n]$. Therefore, any set $S\subseteq [k]\times [n]$  is in $\mathcal{B}$ if and only if both $\det({U_1}[S])$ and $\det({U_2}[S])$ are non-zero. In other words, $\mathcal{B}$ is the set of common bases of $\mathcal{M}_{k,[n]}$ and $\mathcal{M}_{U_2}$. If $\mathcal{B}_{k,[n]}$ and $\mathcal{B}_{U_2}$ are the set of bases of $\mathcal{M}_{k,[n]}$ and $\mathcal{M}_{U_2}$, then 
\begin{equation}\label{eq:Bcommonbase}
    \mathcal{B}= \mathcal{B}_{k,[n]}\cap \mathcal{B}_{U_2}.
\end{equation}

    
    Let $z$ be another variable. If $w\in\mathbb Z_{\geq0}^{[k]\times[n]}$ is an isolating weight for $\mathcal{M}_{k,[n]},\mathcal{M}_{U_2}$, then there exists a unique set in $\mathcal{B}$ that maximizes $w$.  Hence, if we substitute $x_{i,j}$ with $z^{w_{i,j}}$ in $\det({U_1}X{U_2}^{\mathsf T})$, the monomial corresponding to the unique maximum weight set does not get canceled out by any other monomial. Hence, $\det({U_1}X{U_2}^{\mathsf T})$  evaluates to a non-zero univariate polynomial in $z$ of degree equal to the maximum weight of any set in $\mathcal{B}$ with respect to $w$.
    
    For $\widetilde{w}\in \Z^n$, let $w\in \Z^{[k]\times [n]}$ such that \[w_{i,j}=(i-1)\widetilde{w}_j ~~\forall (i,j)\in [k]\times [n].\]
    %
    To obtain an isolating family of the required form, we first
construct weights with nonzero circulation on prescribed
collections of simple cycles. The \cref{lem:contructU} provides
such a construction.
    Recall \cref{def:cycle} for the set of cycles $\mathcal{C}_F$ corresponding to a face $F$ of the matroid intersection polytope. First, we give the following claim
 that states a property of simple cycles corresponding to the matroid intersection polytope when one of the two matroids is the partition matroid $\mathcal{M}_{k,[n]}$.
\begin{claim} \label{obs:simplecyclepartitionmatroid}
    Let $k\leq n$ be two positive integers and $\mathcal{M}$ be any matroid on the ground set $[k]\times [n]$ of rank $n$ having a common base with $\mathcal M_{k,[n]}$. Let $\mathcal{Q}=(Q_1,\cdots ,Q_n)$ be the partition of $[k]\times [n]$ with $Q_j=[k]\times \{j\}$ for each $j\in [n]$.  Let $F$ be a face of $P(\mathcal{M}_{k,[n]},\mathcal{M})$ and  \[C=(e_1,e_2,\dots,e_{2l})\]
    be any simple cycle of $\mathcal{C}_{F}$. Then, there exists an injective map $\phi_C: [l]\to [n]$ such that 
\[
e_{2i-1},e_{2i}\in Q_{\phi_C(i)}
\qquad\text{for every }i\in[l].
\]
\end{claim}
\begin{proof}
For each $j\in [n]$, let $Q_j'$ be set of elements $e$ from $Q_j$ such that $x_e=0$ for each $x\in F$. Let $\mathcal{P}_1$ be the partition of $E=[k]\times [n]$ from \cref{lem:FaceMI} for $\mathcal{M}_1=\mathcal{M}_{k,[n]}$, $\mathcal{M}_2=\mathcal{M}$ and face $F$. Let $S\in \mathcal{P}_1$ with $\nu_{S}>0$. We first show that $S= Q_j\setminus Q_j'$ for some $j\in [n]$. From \cref{lem:FaceMI}, for each element $e$ in $\cup_{j\in[n] }Q_j' $, we get that $\{e\}\in \mathcal{P}_1$ with $\nu_{\{e\}}=0$. Hence, $S$ does not contain any element from  $\cup_{j\in[n] }Q_j' $.

Let $j$ be some element in $[n]$ such that  $S\cap Q_j\neq \emptyset$. Note that \[x(Q_j)=1=r_1(Q_j) \text{ for each } x\in F\subseteq P(\mathcal{M},\mathcal{M}_{k,[n]}).\] From \cref{lem:FaceMI}, $Q_j$ is a union of some sets from $\mathcal{P}_1$ along with $S$. This implies $S\subseteq Q_j\setminus Q_j'$. Since $\nu_{S}$ is at least 1 and $x(Q_j\setminus Q_j')=1$ for each $x\in F$, $S$ is exactly equal to $Q_j\setminus Q_j'.$ Hence, 
\begin{equation}\label{eq:PartForPartitionMatroid}\mathcal{P}_1=\{\{e\}\mid e\in Q'_j, j\in [n]\}\cup \{Q_j\setminus Q'_j\mid j\in [n]\}.\end{equation}

Let $C$ be any simple cycle of $\mathcal{C}_{F}$. From \eqref{eq:PartForPartitionMatroid} and the definition of a simple cycle, $C$ has either two or zero elements from $Q_j\setminus Q_j'$ for each $j\in [n]$. Hence,  we can partition $C$ into pairs, with each pair containing two consecutive elements of $C$, such that both consecutive elements of any pair belong to the same part of $\mathcal{Q}$. Moreover, each of these parts of $\mathcal{Q}$ is unique.
\end{proof}
We now show \cref{lem:contructU}.

    \begin{lemma}\label{lem:contructU}
            Let $n,k, s$ be three positive integers with $k\leq n$. Then, there exists $m=\tilde{O}(nks)$ such that we can construct $\mathcal{U}\subseteq \{0,1,\dots,m\}^{[k]\times [n]}$ of size $\tilde{O}(ns)$ in $\poly(n,k,s)$ with the following properties.
            \begin{enumerate}
                \item For each $u\in \mathcal{U}$, there exists $\widetilde{u}\in \Z_{\geq0}^n$ such that \[u_{i,j}=(i-1)\widetilde{u}_{j} ~~\forall (i,j)\in [k]\times [n].\]
                \item Let $\mathcal{M}$ be any matroid on $[k]\times [n]$ having a common base with $\mathcal M_{k,[n]}$. Let $F$ be any face of  $P(\mathcal{M}_{k,[n]},\mathcal{M})$ and $\mathcal{C}$ be any collection of at most $s$ simple cycles of $\mathcal{C}_{F}$. Then, there exists $u\in \mathcal{U}$ such that \[c_u(C)\neq 0 ~~\forall C\in \mathcal{C}.\]
            \end{enumerate}
    \end{lemma}
   \begin{proof}
        If $k=1$, using \cref{lem:FaceMI}, we can argue that $\mathcal{C}_F=\emptyset$ for any face $F$ of $P(\mathcal{M}_{k,[n]},\mathcal{M})$. Choose $\mathcal{U}=\{0\}$ in this case, and the lemma trivially holds. We therefore assume that $k\geq2$ for the remainder of
the proof. Let $\widetilde{u}\in \Z^n=(1,k,\dots,k^{n-1})$. Let $Q_j=[k]\times \{j\}$ for each $j\in [n]$. Let $C=(e_1,e_2,\dots, e_{2l})$ be any simple cycle of $\mathcal{C}_{F}$ with $e_i\in [k]\times [n]$ for each $i$ in $[2l]$. From \cref{obs:simplecyclepartitionmatroid}, without loss of generality, there exists an injective map $\phi_C: [l]\to [n]$ such that $e_{2i-1}\text{ and }e_{2i}$ both belong to $Q_{\phi_C(i)}$ for each $i\in [l]$. Let $u\in \Z_{\geq0}^{k\times n}$ such that \[u_{i,j}=(i-1)\widetilde{u}_j.\]
        For each $i\in [l]$, let $h_{i_1},h_{i_2}\in [k]$ such that    \[e_{2i-1}=(h_{i_1}, \phi_C(i)) \text{ and } e_{2i}=(h_{i_2}, \phi_C(i)).\]
        Then,
        \begin{align*}
            c_u(C)&=\bigl\vert\sum_{i\in [l]}(u_{e_{2i-1}}-u_{e_{2i}})\bigr\vert\\
            &= \bigl\vert\sum_{i\in [l]} (h_{i_1}-h_{i_2}) \widetilde{u}_{\phi_C(i)}\bigr\vert=\bigl\vert\sum_{i\in [l]} (h_{i_1}-h_{i_2}) k^{(\phi_C(i)-1)}\bigr\vert. 
        \end{align*}
        Since $\phi_C$ is injective and
$0<|h_{i_1}-h_{i_2}|<k$ for every $i\in[l]$, the term involving the
largest power of $k$ cannot be canceled by the remaining terms.
Therefore, $c_u(C)\neq 0$. This also implies that $c_u(C)<k^n$. Let \[ W=\prod_{C\in \mathcal{C}}c_u(C).\] From above discussion, we get that \[0<W<k^{ns}. \]  Since $k\leq n$,  a standard estimate from number theory, following from the Prime Number Theorem for the Chebyshev function, implies that there exists a positive integer $m'=\tilde{O}(ns)$ such that the least common multiple of $\{1,2,\dots, m'\}$ is greater than $k^{ns}$.  If every $p\in[m']$ divided $W$, then lcm$(1,\ldots,m')$ would divide $W$, contradicting
\[
0<W<k^{ns}<\operatorname{lcm}(1,\ldots,m').
\] Hence, there exists some $p\in [m']$ such that $W \mod p\neq 0$.
        
        Let $\mathcal{U}$ be the following set of $m'$ weight functions
        \[\mathcal{U}=\{u^p\in \Z_{\geq0}^{[k]\times [n]}\mid p\in [m'], u^p_{i,j}=(i-1)(k^{j-1}\mod p) \forall i\in [k],j\in [n]   \}.\] 
        The first property of \cref{lem:contructU} directly follows from the above construction of $\mathcal{U}$. Next, we show that $\mathcal{U}$ also satisfies the second property.
        Let $p\in [m']$ such that $W\mod p$ is non-zero.
Since $u^p(e)\equiv u(e)\pmod p$ for every
$e\in[k]\times[n]$, for every cycle $C\in\mathcal C$ we have
\[
c_u(C)\equiv \pm c_{u^p}(C)\pmod p.
\]
In particular,
\[
c_u(C)\not\equiv0\pmod p
\quad\Longleftrightarrow\quad
c_{u^p}(C)\not\equiv0\pmod p.
\]
Since $W\not\equiv0\pmod p$, we have
$c_u(C)\not\equiv0\pmod p$ for every $C\in\mathcal C$.
Consequently, $c_{u^p}(C)\neq0$ for every $C\in\mathcal C$.
        Note that \[
        |\mathcal{U}|\leq m' \text{ and } u_{i,j}\leq km' ~~\forall u\in \mathcal{U}, (i,j)\in [k]\times [n].
        \]
        This completes the proof of \cref{lem:contructU}.
    \end{proof}

    By combining \cref{lem:contructU,lem:isolationset} together, we get an isolating set of weight functions with the required properties in the following corollary.
        \begin{corollary} \label{lem:DesiredIsolationSet}
        Let $k\leq n$ be two positive integers. Then, we can construct in $\qp(nk)$ time a set of weight vectors $\widetilde{\mathcal{W}}\subseteq \Z_{\geq0}^n$ of size $\qp(nk)$ with each coordinate of any  weight vector in $\widetilde{\mathcal{W}}$ bounded by $\qp(nk)$   such that $\mathcal{W}\subseteq \Z^{[k]\times [n]}$ defined as 
        \[\mathcal{W}=\{w\mid w_{i,j}=(i-1)\widetilde{w}_j~ \forall i,j\in [k]\times [n] \text{ and } \widetilde{w}\in \widetilde{\mathcal{W}}\},\]
        is an isolating set for $\mathcal{M},M_{k,[n]}$ for any matroid $\mathcal{M}$ defined on ground set $[k]\times [n]$ having a common base with $\mathcal M_{k,[n]}$.
    \end{corollary}
    \begin{proof}
    If $k=1$, take $\widetilde{\mathcal W}=\{0\}$.
The only possible common base is the entire ground set.
We henceforth assume $k\geq2$.

    Let $E=[k]\times [n]$ and $s=(kn)^4$. Let $\mathcal{U}\subseteq \Z_{\geq0}^{[k]\times[n]}$ be the set of weight vectors from \cref{lem:contructU} for $n,k,s$. From the construction of $\mathcal{U},$ for any weight vector $u\in \mathcal{U}\subseteq \Z_{\geq0}^{[k]\times [n]}$, there exists  $\widetilde{u}\in \Z_{\geq0}^n$ such that 
  $u_{i,j}=(i-1)\widetilde{u}_{j}$ for all $(i,j)\in [k]\times [n]$.

    Let
\[
\ell=\lceil\log_2|E|\rceil-1
\text{ and }
N=1+|E|\max_{u\in\mathcal U,e\in E}|u_e|.
\]  Let $\widetilde{\mathcal{W}}\subseteq \Z_{\geq0}^{n}$ be the set of weight vectors defined such that 
    \[\widetilde{\mathcal{W}}= \{ N^{\ell} \widetilde{u^{0}} + N^{\ell-1}\widetilde{u^1} +\dots+ N^0\widetilde{u^{\ell}} \mid u^i\in \mathcal{U} \text{ for each } i \in \{0,1,\dots,\ell\}\} \]
    Let $\mathcal{W}\subseteq \Z^{[k]\times [n]}$ be the set
        \[\mathcal{W}=\{w\mid w_{i,j}=(i-1)\widetilde{w}_j~ \forall i,j\in [k]\times [n] \text{ and } \widetilde{w}\in \widetilde{\mathcal{W}}\}.\]
From the construction of $\widetilde{\mathcal{W}}$ and $\mathcal{W}$, $\mathcal{W}$ is also equal to the set \[ \{ N^{\ell} u^{0} + N^{\ell-1}u^1 +\dots+ N^0u^{\ell} \mid u^i\in \mathcal{U} \text{ for each } i \in \{0,1,\dots, \ell\}\}.\]
    
    From \cref{lem:contructU},  for any face $F$ of  $P(\mathcal{M}_{k,[n]},\mathcal{M})$ and any collection $\mathcal{C}$   of at most  $|E|^4$ simple cycles of $\mathcal{C}_{F}$, there exists $u\in \mathcal{U}$ such that $c_u(C)\neq 0 \text{ for each cycle } C\in \mathcal{C}.$ Hence, from \cref{lem:isolationset}, $\mathcal{W}$ is an isolating set for $\mathcal{M}_{k,[n]}$ and $\mathcal{M}$.
     \end{proof}
    Next, we use \cref{lem:DesiredIsolationSet} to construct a bivariate polynomial map that preserves the non-zeroness of $\det({U_1}X{U_2}^{\mathsf T})$.  From \eqref{eq:monsupdef} and \eqref{eq:Bcommonbase}, the set of support of monomials of $\det({U_1}X{U_2}^{\mathsf T})$, denoted by $\mathcal{B}$, is the set of common bases of $\mathcal{M}_{k,[n]}$ and $\mathcal{M}_{U_2}$. Let $\widetilde{\mathcal{W}}\subseteq \Z^{n}$ be the set of weight vectors  from \cref{lem:DesiredIsolationSet}. 
    Let \[D_y=|\widetilde{\mathcal{W}}|\text{ and } D_z=\max_{\widetilde{w}\in \widetilde{\mathcal{W}},j\in [n]}\widetilde{w}_j.\] 
    Set $C=\max\{D_y,n\}$,
and let $q\geq C$ be a prime power and $S =\{a_1,a_2\dots, a_{D_y}\}$ be  a subset of $\mathbb{F}_q$.
    Let $\widetilde{w^1},\dots,\widetilde{w^{D_y}}$ be the weight vectors of $\widetilde{\mathcal{W}}$.
    Let $\cG=(\cG_1,\dots,\cG_n)\in\mathbb F_q[y,z]^n$ satisfy
\[
  \cG_j(a_p,z)=z^{\widetilde{w^p}_j}
\]
for all $p\in[D_y]$ and $j\in[n]$.
    The polynomial map $\cG$ satisfying the above condition can be easily found using Lagrange interpolation as follows.
    \[
\cG_j(y,z)
   =\sum_{s=1}^{D_y}L_s(y)z^{\widetilde{w^{s}}_j},
\text{ where }
L_s(y)
  =\prod_{\substack{t\in[D_y]\\t\ne s}}
       \frac{y-a_t}{a_s-a_t}.
\]
Hence,\begin{equation}\label{eq:degG}\deg_y(\cG_j)\leq D_y\text{ and }\deg_z(\cG_j)\leq D_z ~\forall j\in [n].\end{equation}
    Let $f(y,z)$ be the bivariate polynomial obtained by substituting \begin{equation}\label{eq:bisub}
    x_{i,j} \text{ with } (\mathcal{G}_j(y,z))^{i-1} \text{ for each } (i,j)\in [k]\times [n] \text{ in the polynomial }  \det({U_1}X{U_2}^{\mathsf T}).
    \end{equation}
    From \cref{lem:DesiredIsolationSet}, there exists $p\in[D_y]$ such that the corresponding weight vector $w\in \mathbb{Z}^{[k]\times [n]}$ with \[w_{i,j}= (i-1)\widetilde{w^p}_j \text{ for each } (i,j)\in [k]\times [n]\] is isolating for $\mathcal{B}$. For any set $S\in \mathcal{B}$, the substitution of \eqref{eq:bisub} with $y=a_p$ maps the monomial
    \[\prod_{(i,j)\in S}x_{i,j}\to \prod_{(i,j)\in S}(z^{\widetilde{w^p}_j})^{(i-1)}=\prod_{(i,j)\in S}z^{w_{i,j}}=z^{w(S)}.\]
Since $w$ is isolating for $\mathcal{B}$, $f(a_p,z)$ is non-zero which in turn implies that $f(y,z)$ is non-zero.

From \eqref{eq:degG} and \eqref{eq:bisub}, we get $\deg_y(f)\leq nkD_y$ and $\deg_z(f)\leq nkD_z$. Let $t$ be another variable. Now we apply the Kronecker substitution 
$y\mapsto t$ and $z\mapsto t^{nkD_y+1}$.
It preserves nonzeroness of $f$ and yields a univariate polynomial of degree strictly less than
\begin{equation} \label{eq:Dbound}
D
=
(nkD_y+1)(nkD_z+1).
\end{equation}
This is again quasipolynomial in $nk$.
Therefore, if $r\geq D$, then an element
$\beta\in\mathbb{F}_{q^r}$ of degree $r$ over $\mathbb{F}_q$
simultaneously preserves the nonzeroness of $f$ for every pair of
$n\times kn$ matrices $U_1,U_2$. Hence, the elements
$\alpha_1,\ldots,\alpha_n$ defined by $\alpha_j=\cG_j(\beta,\beta^{nkD_y+1})$
satisfy the conclusion of \cref{thm:RSreadonce} for $J=[n]$.

We now show that the same point works for every $J\subseteq[n]$.
Fix $J\subseteq[n]$ and a pair of $|J|\times k|J|$ matrices
$U_1,U_2$ as in the theorem, and let
$\overline{J}:=[n]\setminus J$.

Let $P_1$ be the canonical representation of
$\mathcal{M}_{k,\overline{J}}$. Thus, for every
$j\in\overline{J}$, all the columns of $P_1$ indexed by
$[k]\times\{j\}$ are equal to the $j$th standard basis vector.
Let $P_2$ be the $|\overline{J}|\times k|\overline{J}|$ matrix
whose columns are indexed by $[k]\times\overline{J}$ such that,
for every $j\in\overline{J}$, the column indexed by $(1,j)$ is
the $j$th standard basis vector and all the remaining columns
are zero.

After ordering the columns by first listing $[k]\times J$ and
then $[k]\times\overline{J}$, define
$\widehat{U}_1:=U_1\oplus P_1$ and
$\widehat{U}_2:=U_2\oplus P_2$. The column matroid of
$\widehat{U}_1$ is $\mathcal{M}_{k,[n]}$, and
\[
\widehat{U}_1X\widehat{U}_2^{\mathsf T}
=
\begin{bmatrix}
U_1X_JU_2^{\mathsf T} & 0\\
0 & \operatorname{diag}(x_{1,j}:j\in\overline{J})
\end{bmatrix}.
\]
Consequently,
$\det(\widehat{U}_1X\widehat{U}_2^{\mathsf T})
=
\det(U_1X_JU_2^{\mathsf T})\prod_{j\in\overline{J}}x_{1,j}$.

Thus, if $\det(U_1X_JU_2^{\mathsf T})$ is not the zero polynomial, then
$\det(\widehat{U}_1X\widehat{U}_2^{\mathsf T})$ is not the zero polynomial.
By the case $J=[n]$ proved above, the latter determinant evaluates
to a nonzero element under the substitution
$x_{i,j}=\alpha_j^{i-1}$. Since $x_{1,j}$ evaluates to
$\alpha_j^0=1$ for every $j\in\overline{J}$, it follows that
$\det(U_1X_JU_2^{\mathsf T})$ also evaluates to a nonzero element under the
same substitution. Since $C=\max\{D_y,n\}$,  and from the definition of $D$ (Eq. \eqref{eq:Dbound}), we get
$C,D\leq(nk)^{O(\log(nk))}$.

It remains to show that $\alpha_1,\ldots,\alpha_n$ are pairwise
distinct. Suppose first that $k\geq 2$. Fix distinct $a,b\in[n]$ and
let $J=\{a,b\}$. Let $U_1$ be the canonical representation of
$\mathcal{M}_{k,J}$. Let $\mathbf{e}_{(i,j)}$ denote the standard
basis vector indexed by $(i,j)\in[k]\times J$, and define $U_2$ to
have rows
$\mathbf{e}_{(1,a)}+\mathbf{e}_{(1,b)}$
and
$\mathbf{e}_{(2,a)}+\mathbf{e}_{(2,b)}$.
Then,
\[
U_1X_JU_2^{\mathsf T}
=
\begin{bmatrix}
x_{1,a} & x_{2,a}\\
x_{1,b} & x_{2,b}
\end{bmatrix},
\]
and hence $\det(U_1X_JU_2^{\mathsf T})
=
x_{1,a}x_{2,b}-x_{2,a}x_{1,b}$,
which is not the zero polynomial. By the property proved above,
\[
0
\neq
\left.
\det(U_1X_JU_2^{\mathsf T})
\right|_{x_{i,j}=\alpha_j^{i-1}}
=
\alpha_b-\alpha_a.
\]
Therefore, $\alpha_a\neq\alpha_b$. Since $a$ and $b$ were arbitrary,
the elements $\alpha_1,\ldots,\alpha_n$ are pairwise distinct.

If $k=1$, the substitution is $x_{1,j}=\alpha_j^0=1$ and is therefore
independent of the choice of $\alpha_j$. Thus, we may choose
$\alpha_1,\ldots,\alpha_n$ to be any $n$ distinct elements of
$\mathbb{F}_{q^r}$. 

Finally, we note that the above construction is deterministic.  The weight family from
\cref{lem:DesiredIsolationSet} can be constructed in
$\qp(nk)$ time. Since $D_y$ and $D_z$ are quasipolynomially
bounded in $nk$, the polynomials $G_1,\ldots,G_n$ can be computed
using Lagrange interpolation with $\qp(nk)$ arithmetic operations
over $\mathbb F_q$.
Using a standard deterministic algorithm, we construct an
irreducible polynomial of degree $r$ over $\mathbb F_q$ in
$\poly(qr)$ time. This gives a representation of
$\mathbb F_{q^r}$ and an element $\beta$ of degree $r$ over
$\mathbb F_q$. We then compute
$\alpha_j=G_j(\beta,\beta^{nkD_y+1})$ for every $j\in[n]$ using
$\qp(nk)$ arithmetic operations over $\mathbb F_{q^r}$.
Each such operation takes $\poly(r\log q)$ time.
Thus, the entire construction takes
$\qp(nk)\cdot\poly(qr)$ time.


\section{Proof of \texorpdfstring{\cref{thm:beyondJohnson}}{Theorem~1}}\label{sec:RSLD}
In this section, we give a deterministic construction of Reed-Solomon codes that are list decodable beyond the Johnson bound. Precisely,
we present a proof of \cref{thm:beyondJohnson}. For this, we will use the result of \cite{DBLP:journals/siamcomp/GuoLSTW24}, which gives a sufficient condition on the set of evaluation points. 

To describe this condition, we present some definitions and notations. First, we define a family of matrices $\mathcal{B}_t$, with the property that the rows of $\mathcal{B}_t$ form a basis of the cycle space of the complete graph $K_t$.

\paragraph*{The cycle-space matrix}
For an integer $t\geq 2$, let $\mathcal{B}_t$ be the
$\binom{t-1}{2}\times\binom{t}{2}$ matrix whose rows and columns
are indexed by the two-element subsets of $[t-1]$ and $[t]$,
respectively.\footnote{For $t=2$, $\mathcal{B}_t$ is the $0\times 1$ empty matrix.}
For $1\leq u<v<t$ and $1\leq p<q\leq t$, define
\[
\mathcal{B}_t[\{u,v\},\{p,q\}]
=
\begin{cases}
  1  & \text{if } \{p,q\}=\{u,v\}
       \text{ or } \{p,q\}=\{v,t\},\\
  -1 & \text{if } \{p,q\}=\{u,t\},\\
  0  & \text{otherwise.}
\end{cases}
\]
For example,
\[
\mathcal{B}_4=
\begin{bmatrix}
  1 & 0 & 0 & -1 &  1 & 0\\
  0 & 1 & 0 & -1 &  0 & 1\\
  0 & 0 & 1 &  0 & -1 & 1
\end{bmatrix},
\]
where the rows are indexed by
$\{1,2\}$, $\{1,3\}$, $\{2,3\}$, and the columns are indexed by
$\{1,2\}$, $\{1,3\}$, $\{2,3\}$, $\{1,4\}$, $\{2,4\}$, $\{3,4\}$,
in this order.


\paragraph*{$t$-Wise Intersection Matrices}
We need the notion of intersection matrices from \cite{ST23,DBLP:journals/siamcomp/GuoLSTW24}.

\begin{definition}[{Intersection matrix \cite{ST23,DBLP:journals/siamcomp/GuoLSTW24}}]
    For a $k\times n$ matrix $V$ and $t\geq 3$ subsets $I_1,I_2,\dots,I_t\subseteq [n]$, the corresponding \emph{t-wise intersection} matrix, denoted by $M_{V,(I_1,\dots,I_t)}$, is the following $\left(\binom{t-1}{2}k+\sum_{1\leq i<j\leq t}|I_i\cap I_j|\right)\times \binom{t}{2}k$ matrix:
    \[M_{V,(I_1,\dots,I_t)}=\begin{bmatrix}
        \mathcal{B}_t\otimes \mathcal{I}_k\\
        \hline
        \diag\left(V[I_i\cap I_j]^{\mathsf T}: \{i,j\}\in \binom{[t]}{2} \right)
    \end{bmatrix}.\]
\end{definition}

For $t$ subsets $I_1,\dots, I_t$ of $[n]$, let  \[
\wt(I_1,\dots, I_t)=\left(\sum_{i\in [t]}|I_i|\right)- |\cup_{i\in [t]} I_i|.
\]
For a subset $J\subseteq [t]$, let $\wt(I_J)=\wt(\{I_i\mid i\in J\})$. 
Now, we record the following statement, which follows from \cite[Proof of Theorem 5.1]{DBLP:journals/siamcomp/GuoLSTW24}. It provides a sufficient condition for an RS code to be $(1-\e,O(1/\e))$ list decodable in terms of intersection matrices defined above.

For $t'=2$, we interpret
$M_{V,(I_1,I_2)}$ as the matrix $V[I_1\cap I_2]^{\mathsf T}$.

\begin{theorem}\label{lem:EvalProp}
     There exists an absolute constant $c_0>0$ such that the following
holds.  Let $k,n,t\in \mathbb{N}\text{ and }\e\in (0,1]$ such that
\[
1\leq k\leq
\frac{c_0\e n}{\log(1/\e)+1},
\]
and let $t=\lceil 2/\e\rceil$. Let $x_{i,j}$, $(i,j)\in[k]\times[n]$, be distinct variables and $V$ be the $k\times n$ matrix with $V[i,j]=x_{i,j}$. Let $\alpha_1,\dots ,\alpha_n\in \mathbb{F}$ be pairwise distinct such that for every $2\leq t'\leq t$ and $t'$ 
 subsets $I_1,\dots,I_{t'}$ of $[n]$ satisfying the following conditions
    \begin{gather*} \wt(I_{[t']})=(t'-1)k, ~\wt(I_J)\leq (|J|-1)k  \text{ for each non-empty } J\subseteq [t']\text{ and }\\ I_i\cap I_j\cap I_l=\emptyset  \text{ for distinct } i,j,l\in [t'],\end{gather*}  $M_{V,(I_1,\dots,I_{t'})}$ is full rank over $
\mathbb F(x_{i,j}:(i,j)\in[k]\times[n])
$ if and only if it is full rank over $\mathbb{F}$ on substituting \[x_{i,j}\text{ with }\alpha_j^{i-1} \text{ for each }(i,j)\in [k]\times [n].\] Then,
    the RS code with evaluation points at $\alpha_1,\dots,\alpha_n$ has rate $R=k/n$ and is $(1-\e,t-1)$ list decodable.
\end{theorem}


The following lemma states that the determinant of every intersection matrix
appearing in \cref{lem:EvalProp} is of the form considered in \cref{thm:RSreadonce}, which in turn helps us to construct the evaluation points directly using it. 
\begin{lemma}
\label{lem:intersection-matrix-factorization}
Let $\mathbb{F}$ be a field, and let $I_1,\ldots,I_{t'}\subseteq[n]$
satisfy
\[
\wt(I_{[t']})=(t'-1)k
\text{ and }
I_a\cap I_b\cap I_c=\emptyset
~
\text{for all distinct }a,b,c\in[t'].
\]
Let $V$ be the $k\times n$ matrix defined by
$V[i,j]=x_{i,j}$
for every $(i,j)\in[k]\times[n].$

Then, there exists a subset $J\subseteq[n]$ and matrices
$
U_1,U_2\in\mathbb{F}^{|J|\times k|J|}
$
whose columns are indexed by $[k]\times J$ such that the column
matroid represented by $U_1$ is $\mathcal{M}_{k,J}$ and
\[
\det\left(M_{V,(I_1,\ldots,I_{t'})}\right)
=
\det(U_1X_JU_2^{\mathsf T}).
\]
\end{lemma}
Here, we give a proof of \cref{lem:intersection-matrix-factorization}.
\begin{proof}
For $1\leq a<b\leq t'$, let $J_{a,b}:=I_a\cap I_b$, and let
\[
J:=\bigcup_{1\leq a<b\leq t'}J_{a,b}.
\]
Since the intersection of any three sets from
$I_1,\ldots,I_{t'}$ is empty, the sets $J_{a,b}$ are pairwise
disjoint. Therefore,
\[
|J|
=
\sum_{1\leq a<b\leq t'}|I_a\cap I_b|
=
\wt(I_{[t']})
=
(t'-1)k.
\]

Let
\[
s:=k\binom{t'-1}{2}
\qquad\text{and}\qquad
d:=(t'-1)k=|J|.
\]
Order the columns of $M_{V,(I_1,\ldots,I_{t'})}$ by first listing
those indexed by $\binom{[t'-1]}{2}\times[k]$, followed by those
indexed by
\[
\bigl\{(a,t'):a\in[t'-1]\bigr\}\times[k].
\]
By the definition of $\mathcal{B}_{t'}$, the first
$\binom{t'-1}{2}$ columns of $\mathcal{B}_{t'}$ form an identity
matrix. Hence, for some matrices $A_0,B_0$, and $C$, we can write
\[
M_{V,(I_1,\ldots,I_{t'})}
=
\begin{bmatrix}
I_s & C\\
A_0 & B_0
\end{bmatrix}.
\]
Performing determinant-preserving column operations gives
\[
\begin{bmatrix}
I_s & C\\
A_0 & B_0
\end{bmatrix}
\begin{bmatrix}
I_s & -C\\
0   & I_d
\end{bmatrix}
=
\begin{bmatrix}
I_s & 0\\
A_0 & B_0-A_0C
\end{bmatrix}.
\]
Let $B:=B_0-A_0C$. Since the matrix used for the column operations
has determinant one, we get
\[
\det\left(M_{V,(I_1,\ldots,I_{t'})}\right)=\det(B).
\]

We next describe the entries of $B$. Let
$\mathcal{E}_0:=\binom{[t'-1]}{2}$ and write
\[
\mathcal{B}_{t'}
=
\begin{bmatrix}
I_{\binom{t'-1}{2}} & C_0
\end{bmatrix}.
\]
By the definition of $\mathcal{B}_{t'}$, for every
$(u,v)\in\mathcal{E}_0$ and $a\in[t'-1]$,
\[
C_0[(u,v),a]
=
\begin{cases}
-1 & \text{if }a=u,\\
 1 & \text{if }a=v,\\
 0 & \text{otherwise}.
\end{cases}
\]
Moreover, $C=C_0\otimes I_k$.

Since the intersection of any three of the sets
$I_1,\ldots,I_{t'}$ is empty, for every $j\in J$, there is a unique
pair $p(j)<q(j)$ such that
\[
j\in I_{p(j)}\cap I_{q(j)}.
\]

\begin{claim}\label{claim:entries-of-B}
Index the rows of $B$ by $J$ and its columns by
$[t'-1]\times[k]$, where $(a,i)$ corresponds to the original column
indexed by $((a,t'),i)$. Then, for every $j\in J$, $a\in[t'-1]$,
and $i\in[k]$,
\[
B[j,(a,i)]
=
\begin{cases}
 x_{i,j} & \text{if }a=p(j),\\
-x_{i,j} & \text{if }a=q(j),\\
 0       & \text{otherwise}.
\end{cases}
\]
\end{claim}

\begin{proof}
Suppose first that $q(j)=t'$. In this case, the $j$th row of $A_0$
is zero, whereas
\[
B_0[j,(a,i)]
=
\begin{cases}
x_{i,j} & \text{if }a=p(j),\\
0       & \text{otherwise}.
\end{cases}
\]
Since $B=B_0-A_0C$, we get
\[
B[j,(a,i)]
=
\begin{cases}
x_{i,j} & \text{if }a=p(j),\\
0       & \text{otherwise}.
\end{cases}
\]
This agrees with the statement of the claim, since
$a\in[t'-1]$ and $q(j)=t'$.

Now suppose that $q(j)<t'$. In this case, the $j$th row of $B_0$
is zero. Furthermore, for every $(u,v)\in\mathcal{E}_0$ and
$i'\in[k]$,
\[
A_0[j,((u,v),i')]
=
\begin{cases}
x_{i',j}
  & \text{if }(u,v)=(p(j),q(j)),\\
0 & \text{otherwise}.
\end{cases}
\]
Since $C=C_0\otimes I_k$, the only nonzero term in the computation
of $(A_0C)[j,(a,i)]$ is the term corresponding to
$(u,v)=(p(j),q(j))$ and $i'=i$. Hence,
\begin{align*}
(A_0C)[j,(a,i)]
&=
x_{i,j}C_0[(p(j),q(j)),a]\\
&=
\begin{cases}
-x_{i,j} & \text{if }a=p(j),\\
 x_{i,j} & \text{if }a=q(j),\\
 0       & \text{otherwise}.
\end{cases}
\end{align*}
Since the $j$th row of $B_0$ is zero and $B=B_0-A_0C$, it follows
that
\[
B[j,(a,i)]
=
\begin{cases}
 x_{i,j} & \text{if }a=p(j),\\
-x_{i,j} & \text{if }a=q(j),\\
 0       & \text{otherwise}.
\end{cases}
\]
This completes the proof of the claim.
\end{proof}

We now write $B$ in the form $U_1X_JU_2^{\mathsf T}$. Define the
$|J|\times k|J|$ matrix $U_1$, whose rows are indexed by $J$ and
whose columns are indexed by $[k]\times J$, by
\[
U_1[j,(i,j')]
=
\begin{cases}
1 & \text{if }j=j',\\
0 & \text{otherwise}.
\end{cases}
\]
For every $j\in J$, the columns indexed by
$[k]\times\{j\}$ are all equal to the $j$th standard basis vector.
Hence, the column matroid represented by $U_1$ is
$\mathcal{M}_{k,J}$.

Define the $|J|\times k|J|$ matrix $U_2$, whose rows are indexed by
$[t'-1]\times[k]$ and whose columns are indexed by $[k]\times J$,
by
\[
U_2[(a,i),(i',j)]
=
\begin{cases}
 1 & \text{if }i=i'\text{ and }a=p(j),\\
-1 & \text{if }i=i'\text{ and }a=q(j),\\
 0 & \text{otherwise}.
\end{cases}
\]
For every $j\in J$, $a\in[t'-1]$, and $i\in[k]$, we have
\begin{align*}
(U_1X_JU_2^{\mathsf T})[j,(a,i)]
&=
\sum_{\substack{i'\in[k]\\j'\in J}}
U_1[j,(i',j')]x_{i',j'}U_2[(a,i),(i',j')]\\
&=
\sum_{i'\in[k]}x_{i',j}U_2[(a,i),(i',j)]\\
&=
\begin{cases}
 x_{i,j} & \text{if }a=p(j),\\
-x_{i,j} & \text{if }a=q(j),\\
 0       & \text{otherwise}.
\end{cases}
\end{align*}
By \cref{claim:entries-of-B}, this is equal to $B[j,(a,i)]$.
Therefore,
\[
B=U_1X_JU_2^{\mathsf T}.
\]
Consequently,
\[
\det\left(M_{V,(I_1,\ldots,I_{t'})}\right)
=
\det(U_1X_JU_2^{\mathsf T}).
\]
This completes the proof of \cref{lem:intersection-matrix-factorization}.
\end{proof}

\begin{proof}[Proof of \cref{thm:beyondJohnson}]
Let $c_0$ be the absolute constant from \cref{lem:EvalProp}, and set
\[
k=
\left\lfloor
\frac{c_0\e n}{\log(1/\e)+1}
\right\rfloor
\text{ and }
t=\left\lceil\frac{2}{\e}\right\rceil.
\]
For sufficiently large $n$, we have $k\geq1$ and
$\frac{k}{n}
=
\Omega\left(
\frac{\e}{\log(1/\e)+1}
\right)$.

Let $C,D\leq(nk)^{O(\log(nk))}$ be the positive integers from
\cref{thm:RSreadonce}. Choose a power of two $q$ such that
\[
\max\{C,n\}\leq q<2\max\{C,n\},
\]
and let $r=D$. By \cref{thm:RSreadonce}, there exist pairwise distinct elements
$\alpha_1,\ldots,\alpha_n\in\mathbb{F}_{q^r}$ satisfying the
conclusion of that theorem.

We show that $(\alpha_1,\ldots,\alpha_n)$ satisfies the
condition of \cref{lem:EvalProp}. Let $t'\leq t$, and let
$I_1,\ldots,I_{t'}\subseteq[n]$ satisfy all the conditions appearing
in \cref{lem:EvalProp}. Suppose that $\det(M_{V,(I_1,\ldots,I_{t'})})$ is not the zero polynomial.

By \cref{lem:intersection-matrix-factorization}, there exist a subset
$J\subseteq[n]$ and matrices
$U_1,U_2\in\mathbb{F}_q^{|J|\times k|J|}$ such that the column
matroid represented by $U_1$ is $\mathcal{M}_{k,J}$ and
\[
\det\left(M_{V,(I_1,\ldots,I_{t'})}\right)
=
\det(U_1X_JU_2^{\mathsf T}).
\]
In particular, $\det(U_1X_JU_2^{\mathsf T})$ is not the zero
polynomial. Hence, by \cref{thm:RSreadonce},
\[
\left.
\det(U_1X_JU_2^{\mathsf T})
\right|_{x_{i,j}=\alpha_j^{\,i-1}}
\neq 0.
\]
Therefore,
\[
\left.
\det\left(M_{V,(I_1,\ldots,I_{t'})}\right)
\right|_{x_{i,j}=\alpha_j^{\,i-1}}
\neq 0.
\]

Conversely, if
$\det(M_{V,(I_1,\ldots,I_{t'})})$ is the zero polynomial, then it
remains zero under every substitution. Thus, $(\alpha_1,\ldots,\alpha_n)$  satisfies the
condition of \cref{lem:EvalProp}.

It follows from \cref{lem:EvalProp} that the RS code with evaluation
points $\alpha_1,\ldots,\alpha_n$ has rate
$R=\frac{k}{n}
=
\Omega\left(
\frac{\e}{\log(1/\e)+1}
\right)$
and is $(1-\e,O(1/\e))$ list decodable.

Finally, $|\mathbb{F}_{q^r}|
=
q^r
\leq
\left(2\max\{C,n\}\right)^D
\leq
(nk)^{\operatorname{quasi\text{-}poly}(nk)}$.
The evaluation points can be constructed in
$\operatorname{quasi\text{-}poly}(nk)$ time by \cref{thm:RSreadonce}. This
completes the proof.
\end{proof}
\section{Proof of \texorpdfstring{\cref{thm:MDS3}}{Theorem~3}}\label{sec:RSMDS3}

In this section, we give a  deterministic construction of MDS$(3)$ Reed--Solomon codes.  For this, we use an equivalent characterization of these codes using the following matrices by \cite{BDG23}. 
\paragraph*{\texorpdfstring{$\ell$}{ell}-MDS Matrices}
\begin{definition}
    For a $k\times n$ matrix $V$ and $\ell\geq 2$ subsets $A_1,A_2,\dots,A_{\ell}\subseteq [n]$, the \emph{$\ell$-MDS} matrix $N_{V,(A_1,\dots,A_{\ell})}$ is the $\ell k\times \left(k+\sum_{i=1}^{\ell}|A_i|\right)$ matrix defined as follows:
    \[N_{V,(A_1,\dots,A_{\ell})}=\begin{bmatrix}
    I_k & V[A_1] & && \\
I_k & & V[A_2] & &\\
\vdots & & & \ddots & \\
I_k &  & & & V[A_{\ell}]
    \end{bmatrix}.\]
\end{definition}

The following theorem gives the equivalent condition for an RS code to be an $[n,k]$-MDS$(\ell)$ code using $\ell$-MDS matrices.
\begin{theorem}[\cite{BGM24}]\label{thm:lMDSCondition}
    Let $k,n,\ell$ be positive integers with $k\leq n$ and $\ell\geq 2$. Let $x_{i,j}$, $(i,j)\in[k]\times[n]$, be distinct variables and $V$ be the $k\times n$ matrix with $V[i,j]=x_{i,j}$. Let $\alpha_1,\dots ,\alpha_n\in \mathbb{F}$ be pairwise distinct such that for any $\ell$ subsets $A_1,A_2,\dots, A_{\ell}$ of $[n]$ with $|A_1|+\cdots +|A_{\ell}|=(\ell-1)k$, $N_{V,(A_1,\dots,A_{\ell})}$ is full rank over $\mathbb F(x_{i,j}:(i,j)\in[k]\times[n])$ if and only if it is full rank over $\mathbb{F}$ on substituting \[x_{i,j}\text{ with }\alpha_j^{i-1} \text{ for each }(i,j)\in [k]\times [n].\] Then,
    the RS code with evaluation points $\alpha_1,\dots,\alpha_n$ is $[n,k]$-MDS$(\ell)$.
   \end{theorem}

   For $\ell=3$, Brakensiek, Dhar, and Gopi~\cite[Corollary 42 and Remark 43]{BDG23} further simplified the condition
and showed that it is sufficient to consider only  $3$-MDS matrices $N_{V,(A_1,A_2,A_3)}$ for pairwise
disjoint sets, possibly for an RS code of smaller dimension. We use
the following consequence of their result. 
\begin{lemma}[\cite{BDG23}]\label{lem:3MDSCondition}
    Let $k\leq n$ be two positive integers. Let $x_{i,j}$, $(i,j)\in[k]\times[n]$, be distinct variables and $V$ be the $k\times n$ matrix with $V[i,j]=x_{i,j}$. For each $k'\in[k]$, let
$V^{(k')}$ denote the submatrix of $V$ consisting of its first $k'$
rows.

Let $\alpha_1,\dots,\alpha_n\in\mathbb{F}$ be pairwise distinct such that,
for every $k'\in[k]$ and every three pairwise disjoint subsets
$A_1,A_2,A_3\subseteq[n]$ satisfying
$|A_i|\leq k'$ for each $i\in[3]$ and $|A_1|+|A_2| +|A_{3}|=2k'$, $N_{V^{(k')},(A_1,A_2,A_{3})}$ is full rank over $\mathbb F(x_{i,j}:(i,j)\in[k']\times[n])$ if and only if it is full rank over $\mathbb{F}$ on substituting \[x_{i,j}\text{ with }\alpha_j^{i-1} \text{ for each }(i,j)\in [k']\times [n].\] Then,
    the RS code with evaluation points $\alpha_1,\dots,\alpha_n$ is $[n,k]$-MDS$(3)$.
   \end{lemma}

By \cref{lem:3MDSCondition}, it suffices to verify the required
condition for every $k'\leq k$. We present the argument for $k'=k$. The following lemma states that the determinant of every $3$-MDS matrix appearing
in \cref{lem:3MDSCondition} is of the form considered in
\cref{thm:RSreadonce}. This, in turn, allows us to directly use \cref{thm:RSreadonce} to construct the evaluation points. 

The case $k'<k$ follows from the same argument together with the
padding idea used in the proof of \cref{thm:RSreadonce}. Indeed, after applying
\cref{lem:3MDS-matrix-factorization} with $k'$ in place of $k$, we
add, in each part indexed by $j\in J$, $k-k'$ parallel columns to
the first matrix and $k-k'$ zero columns to the second matrix. The
first padded matrix represents $\mathcal{M}_{k,J}$, while the added
columns make no contribution to the matrix product. Hence, the
determinant remains unchanged, and \cref{thm:RSreadonce} applies.

\begin{lemma}\label{lem:3MDS-matrix-factorization}
Let $\mathbb{F}$ be a field, and let $A_1,A_2,A_3\subseteq[n]$ be
pairwise disjoint subsets satisfying
$|A_1|+|A_2|+|A_3|=2k$. Let $V$ be the $k\times n$ matrix defined by
$V[i,j]=x_{i,j}$ for every $(i,j)\in[k]\times[n]$. Then, for
$J:=A_1\cup A_2\cup A_3$, there exist matrices
$U_1,U_2\in\mathbb{F}^{|J|\times k|J|}$ whose columns are indexed by
$[k]\times J$ such that the column matroid represented by $U_1$ is
$\mathcal{M}_{k,J}$ and
\[
\det\left(N_{V,(A_1,A_2,A_3)}\right)
=
\det(U_1X_JU_2^{\mathsf T}).
\]
\end{lemma}
\begin{proof}
Since $A_1,A_2,A_3$ are pairwise disjoint, we have
$|J|=|A_1|+|A_2|+|A_3|=2k$. By definition,
\[
N_{V,(A_1,A_2,A_3)}
=
\begin{bmatrix}
I_k & V[A_1] & 0       & 0\\
I_k & 0       & V[A_2] & 0\\
I_k & 0       & 0       & V[A_3]
\end{bmatrix}.
\]
For every $i\in[k]$, subtract the $i$th row from the $(k+i)$th and
$(2k+i)$th rows. These row operations preserve the determinant and
give
\[
\begin{bmatrix}
I_k & V[A_1]  & 0       & 0\\
0   & -V[A_1] & V[A_2] & 0\\
0   & -V[A_1] & 0       & V[A_3]
\end{bmatrix}.
\]
Therefore,
\[
\det\left(N_{V,(A_1,A_2,A_3)}\right)=\det(B),
\]
where
\[
B:=
\begin{bmatrix}
-V[A_1] & V[A_2] & 0\\
-V[A_1] & 0       & V[A_3]
\end{bmatrix}.
\]
We index the rows of $B$ by $[2]\times[k]$ and its columns by $J$,
where the elements of $A_1$, $A_2$, and $A_3$ are listed in this
order.

We now write $B^{\mathsf T}$ in the form
$U_1X_JU_2^{\mathsf T}$. Define the $|J|\times k|J|$ matrix $U_1$,
whose rows are indexed by $J$ and whose columns are indexed by
$[k]\times J$, by
\[
U_1[j,(i,j')]
=
\begin{cases}
1 & \text{if }j=j',\\
0 & \text{otherwise}.
\end{cases}
\]
For every $j\in J$, all the columns indexed by
$[k]\times\{j\}$ are equal to the $j$th standard basis vector.
Hence, the column matroid represented by $U_1$ is
$\mathcal{M}_{k,J}$.

Define the $|J|\times k|J|$ matrix $U_2$, whose rows are indexed by
$[2]\times[k]$ and whose columns are indexed by $[k]\times J$, by
\[
U_2[(a,i),(i',j)]
=
\begin{cases}
-1 & \text{if }i=i'\text{ and }j\in A_1,\\
 1 & \text{if }i=i',\,a=1,\text{ and }j\in A_2,\\
 1 & \text{if }i=i',\,a=2,\text{ and }j\in A_3,\\
 0 & \text{otherwise}.
\end{cases}
\]
For every $j\in J$, $a\in[2]$, and $i\in[k]$, we have
\begin{align*}
(U_1X_JU_2^{\mathsf T})[j,(a,i)]
&=
\sum_{\substack{i'\in[k]\\j'\in J}}
U_1[j,(i',j')]x_{i',j'}U_2[(a,i),(i',j')]\\
&=
\sum_{i'\in[k]}x_{i',j}U_2[(a,i),(i',j)]\\
&=
\begin{cases}
-x_{i,j} & \text{if }j\in A_1,\\
 x_{i,j} & \text{if }a=1\text{ and }j\in A_2,\\
 x_{i,j} & \text{if }a=2\text{ and }j\in A_3,\\
 0       & \text{otherwise}.
\end{cases}
\end{align*}
By the definition of $B$, this is equal to
$B^{\mathsf T}[j,(a,i)]$. Therefore,
\[
B^{\mathsf T}=U_1X_JU_2^{\mathsf T}.
\]
Consequently,
\[
\det\left(N_{V,(A_1,A_2,A_3)}\right)
=
\det(B)
=
\det(B^{\mathsf T})
=
\det(U_1X_JU_2^{\mathsf T}).
\]
\end{proof}

\begin{proof}[Proof of \cref{thm:MDS3}]
Let $C,D\leq(nk)^{O(\log(nk))}$ be the positive integers from
\cref{thm:RSreadonce}. Choose a power of two $q$ such that
\[
\max\{C,n\}\leq q<2\max\{C,n\},
\]
and let $r=D$. By \cref{thm:RSreadonce}, there exist pairwise distinct elements
$\alpha_1,\ldots,\alpha_n\in\mathbb{F}_{q^r}$ satisfying the
conclusion of \cref{thm:RSreadonce}.

We show that $\alpha=(\alpha_1,\ldots,\alpha_n)$ satisfies the
condition of \cref{lem:3MDSCondition}. 
We present the argument for $k'=k$; for $k'<k$, the same argument
applies with $k$ replaced by $k'$, using the padding construction
preceding \cref{lem:3MDS-matrix-factorization}.
Let $A_1,A_2,A_3\subseteq[n]$
be pairwise disjoint subsets satisfying
$|A_1|+|A_2|+|A_3|=2k$, and let $V$ be the $k\times n$ matrix
defined by $V[i,j]=x_{i,j}$.

Suppose that $N_{V,(A_1,A_2,A_3)}$ is full rank over
$\mathbb{F}_{q^r}(x_{i,j}:(i,j)\in[k]\times[n])$. Equivalently,
$\det(N_{V,(A_1,A_2,A_3)})$ is not the zero polynomial. By
\cref{lem:3MDS-matrix-factorization}, for
$J=A_1\cup A_2\cup A_3$, there exist matrices
$U_1,U_2\in\mathbb{F}_q^{|J|\times k|J|}$ such that the column
matroid represented by $U_1$ is $\mathcal{M}_{k,J}$ and
\[
\det\left(N_{V,(A_1,A_2,A_3)}\right)
=
\det(U_1X_JU_2^{\mathsf T}).
\]
In particular, $\det(U_1X_JU_2^{\mathsf T})$ is not the zero
polynomial. Hence, by \cref{thm:RSreadonce},
\[
\left.
\det(U_1X_JU_2^{\mathsf T})
\right|_{x_{i,j}=\alpha_j^{\,i-1}}
\neq 0.
\]
Therefore,
\[
\left.
\det\left(N_{V,(A_1,A_2,A_3)}\right)
\right|_{x_{i,j}=\alpha_j^{\,i-1}}
\neq 0.
\]
Thus, $N_{V,(A_1,A_2,A_3)}$ remains full rank after the substitution
$x_{i,j}=\alpha_j^{i-1}$ for every $(i,j)\in[k]\times[n]$.

Conversely, if $N_{V,(A_1,A_2,A_3)}$ is not full rank over the
rational function field, then its determinant is the zero polynomial
and remains zero under every substitution. Hence, $\alpha$ satisfies
the condition of \cref{lem:3MDSCondition}.

It follows from \cref{lem:3MDSCondition} that the RS code
with evaluation points $\alpha_1,\ldots,\alpha_n$ is an
$[n,k]$-MDS$(3)$ code.

Finally, $|\mathbb{F}_{q^r}|
=
q^r
\leq
\left(2\max\{C,n\}\right)^D
\leq
(nk)^{\operatorname{quasi\text{-}poly}(nk)}$.
The evaluation points can be constructed in
$\operatorname{quasi\text{-}poly}(nk)$ time by \cref{thm:RSreadonce}. This
completes the proof.
\end{proof}

\paragraph*{AI Disclosure:}
All mathematical ideas, results, proofs, and technical arguments in this
paper are due to the authors. Generative AI tools were used to assist with
language editing, correcting grammatical and typographical errors, improving
the presentation and organization of the manuscript, and suggesting clearer
wording for some definitions and proofs. The authors independently checked
and take full responsibility for all mathematical statements and arguments.
\bibliographystyle{alpha}
\bibliography{biblio}

@article{DBLP:journals/siamcomp/GuoLSTW24,
  author       = {Zeyu Guo and
                  Ray Li and
                  Chong Shangguan and
                  Itzhak Tamo and
                  Mary Wootters},
  title        = {Improved List-Decodability and List-Recoverability of {R}eed-{S}olomon
                  Codes via Tree Packings},
  journal      = {{SIAM} J. Comput.},
  volume       = {53},
  number       = {2},
  pages        = {389--430},
  year         = {2024},
  url          = {https://doi.org/10.1137/21m1463707},
  doi          = {10.1137/21M1463707},
  bibsource    = {dblp computer science bibliography, https://dblp.org}
}

@book{Oxl06,
 author = {Oxley, James G.},
 title = {Matroid Theory (Oxford Graduate Texts in Mathematics)},
 year = {2006},
 isbn = {0199202508},
 publisher = {Oxford University Press, Inc.},
 address = {New York, NY, USA},
}

@book{Sch03B,
   title = "Combinatorial optimization : polyhedra and efficiency. Vol. {B}. , Matroids, trees, stable sets. chapters 39-69",
   author = "Schrijver, Alexander",
   series = "Algorithms and combinatorics",
   publisher = "Springer-Verlag",
   address = "Berlin, Heidelberg, New York, N.Y., et al.",
   url = "http://opac.inria.fr/record=b1124843",
   isbn = "3-540-44389-4",
   year = 2003
}

@article{DBLP:journals/cc/GurjarT20,
  author       = {Rohit Gurjar and
                  Thomas Thierauf},
  title        = {Linear Matroid Intersection is in Quasi-{NC}},
  journal      = {Comput. Complex.},
  volume       = {29},
  number       = {2},
  pages        = {9},
  year         = {2020},
  url          = {https://doi.org/10.1007/s00037-020-00200-z},
  doi          = {10.1007/S00037-020-00200-Z},
  bibsource    = {dblp computer science bibliography, https://dblp.org}
}

@article{STV01,
  author  = {Madhu Sudan and Luca Trevisan and Salil Vadhan},
  title   = {Pseudorandom Generators without the {XOR} Lemma},
  journal = {Journal of Computer and System Sciences},
  volume  = {62},
  number  = {2},
  pages   = {236--266},
  year    = {2001},
  doi     = {10.1006/jcss.2000.1730}
}

@article{Singleton64,
  author  = {Richard C. Singleton},
  title   = {Maximum Distance {$q$}-nary Codes},
  journal = {IEEE Transactions on Information Theory},
  volume  = {10},
  number  = {2},
  pages   = {116--118},
  year    = {1964},
  doi     = {10.1109/TIT.1964.1053661}
}

@inproceedings{Elias57,
  author    = {Peter Elias},
  title     = {List Decoding for Noisy Channels},
  booktitle = {IRE WESCON Convention Record, Part 2},
  pages     = {94--104},
  year      = {1957},
  publisher = {Institute of Radio Engineers},
  url       = {https://dspace.mit.edu/handle/1721.1/4484}
}

@techreport{Wozencraft58,
  author      = {John M. Wozencraft},
  title       = {List Decoding},
  institution = {Research Laboratory of Electronics,
                 Massachusetts Institute of Technology},
  type        = {Quarterly Progress Report},
  number      = {48},
  year        = {1958},
}

@inproceedings{GL89,
  author    = {Oded Goldreich and Leonid A. Levin},
  title     = {A Hard-Core Predicate for All One-Way Functions},
  booktitle = {Proceedings of the Twenty-First Annual ACM Symposium on Theory of Computing},
  pages     = {25--32},
  publisher = {ACM},
  year      = {1989},
  doi       = {10.1145/73007.73010}
}

@inproceedings{CPS99,
  author    = {Jin-Yi Cai and Aduri Pavan and D. Sivakumar},
  title     = {On the Hardness of Permanent},
  booktitle = {STACS 99},
  series    = {Lecture Notes in Computer Science},
  volume    = {1563},
  pages     = {90--99},
  publisher = {Springer},
  year      = {1999},
  doi       = {10.1007/3-540-49116-3_8}
}

@article{Tre01,
  author  = {Luca Trevisan},
  title   = {Extractors and Pseudorandom Generators},
  journal = {Journal of the ACM},
  volume  = {48},
  number  = {4},
  pages   = {860--879},
  year    = {2001},
  doi     = {10.1145/502090.502099}
}

@inproceedings{AGS03,
  author    = {Adi Akavia and Shafi Goldwasser and Shmuel Safra},
  title     = {Proving Hard-Core Predicates Using List Decoding},
  booktitle = {Proceedings of the 44th Annual IEEE Symposium on Foundations of Computer Science},
  pages     = {146--157},
  publisher = {IEEE Computer Society},
  year      = {2003},
  doi       = {10.1109/SFCS.2003.1238189}
}

@article{Uma03,
  author  = {Christopher Umans},
  title   = {Pseudo-random Generators for All Hardnesses},
  journal = {Journal of Computer and System Sciences},
  volume  = {67},
  number  = {2},
  pages   = {419--440},
  year    = {2003},
  doi     = {10.1016/S0022-0000(03)00046-1}
}

@article{SSW03,
  author  = {Alice Silverberg and Jessica Staddon and Judy L. Walker},
  title   = {Applications of List Decoding to Tracing Traitors},
  journal = {IEEE Transactions on Information Theory},
  volume  = {49},
  number  = {5},
  pages   = {1312--1318},
  year    = {2003},
  doi     = {10.1109/TIT.2003.810630}
}

@article{TZ04,
  author  = {Amnon Ta-Shma and David Zuckerman},
  title   = {Extractor Codes},
  journal = {IEEE Transactions on Information Theory},
  volume  = {50},
  number  = {12},
  pages   = {3015--3025},
  year    = {2004},
  doi     = {10.1109/TIT.2004.838377}
}

@article{SU05,
  author  = {Ronen Shaltiel and Christopher Umans},
  title   = {Simple Extractors for All Min-Entropies and a New Pseudorandom Generator},
  journal = {Journal of the ACM},
  volume  = {52},
  number  = {2},
  pages   = {172--216},
  year    = {2005},
  doi     = {10.1145/1059513.1059516},
}

@article{TZS06,
  author  = {Amnon Ta-Shma and David Zuckerman and Shmuel Safra},
  title   = {Extractors from {Reed--Muller} Codes},
  journal = {Journal of Computer and System Sciences},
  volume  = {72},
  number  = {5},
  pages   = {786--812},
  year    = {2006},
  doi     = {10.1016/j.jcss.2005.05.010},
}

@article{IJK09,
  author  = {Russell Impagliazzo and Ragesh Jaiswal and Valentine Kabanets},
  title   = {Approximate List-Decoding of Direct Product Codes and Uniform Hardness Amplification},
  journal = {SIAM Journal on Computing},
  volume  = {39},
  number  = {2},
  pages   = {564--605},
  year    = {2009},
  doi     = {10.1137/070683994}
}

@article{Uma09,
  author  = {Christopher Umans},
  title   = {Reconstructive Dispersers and Hitting Set Generators},
  journal = {Algorithmica},
  volume  = {55},
  number  = {1},
  pages   = {134--156},
  year    = {2009},
  doi     = {10.1007/s00453-008-9266-z},
}

@article{GUV09,
  author  = {Venkatesan Guruswami and Christopher Umans and Salil Vadhan},
  title   = {Unbalanced Expanders and Randomness Extractors from {Parvaresh--Vardy} Codes},
  journal = {Journal of the ACM},
  volume  = {56},
  number  = {4},
  pages   = {1--34},
  year    = {2009},
  doi     = {10.1145/1538902.1538904}
}

@article{ST23,
  author  = {Chong Shangguan and Itzhak Tamo},
  title   = {Generalized {Singleton} Bound and List-Decoding
             {Reed--Solomon} Codes Beyond the {Johnson} Radius},
  journal = {SIAM Journal on Computing},
  volume  = {52},
  number  = {3},
  pages   = {684--717},
  year    = {2023},
  doi     = {10.1137/20M138795X},
}

@article{Rot22,
  author  = {Ron M. Roth},
  title   = {Higher-Order {MDS} Codes},
  journal = {IEEE Transactions on Information Theory},
  volume  = {68},
  number  = {12},
  pages   = {7798--7816},
  year    = {2022},
}

@article{Joh62,
  author  = {Selmer M. Johnson},
  title   = {A New Upper Bound for Error-Correcting Codes},
  journal = {IRE Transactions on Information Theory},
  volume  = {8},
  number  = {3},
  pages   = {203--207},
  year    = {1962},
  doi     = {10.1109/TIT.1962.1057714}
}

@unpublished{GS01,
  author = {Venkatesan Guruswami and Madhu Sudan},
  title  = {Extensions to the {Johnson} Bound},
  note   = {Manuscript},
  year   = {2001},
  url    = {https://people.csail.mit.edu/madhu/papers/2001/johnson.pdf}
}

@inproceedings{RW14,
  author    = {Atri Rudra and Mary Wootters},
  title     = {Every List-Decodable Code for High Noise Has
               Abundant Near-Optimal Rate Puncturings},
  booktitle = {Proceedings of the 46th Annual ACM Symposium
               on Theory of Computing},
  pages     = {764--773},
  year      = {2014},
  publisher = {ACM}
}

@article{FKS22,
  author  = {Asaf Ferber and Matthew Kwan and Lisa Sauermann},
  title   = {List-Decodability with Large Radius
             for {Reed--Solomon} Codes},
  journal = {IEEE Transactions on Information Theory},
  volume  = {68},
  number  = {6},
  pages   = {3823--3828},
  year    = {2022}
}

@article{GST23,
  author  = {Eitan Goldberg and Chong Shangguan and Itzhak Tamo},
  title   = {List-Decoding and List-Recovery of {Reed--Solomon}
             Codes Beyond the {Johnson} Radius for Every Rate},
  journal = {IEEE Transactions on Information Theory},
  volume  = {69},
  number  = {4},
  pages   = {2261--2268},
  year    = {2023},
  doi     = {10.1109/TIT.2022.3222877}
}

@article{BGM24,
  author  = {Joshua Brakensiek and Sivakanth Gopi and Visu Makam},
  title   = {Generic {Reed--Solomon} Codes Achieve List-Decoding Capacity},
  journal = {SIAM Journal on Computing},
  pages   = {STOC23-118--STOC23-154},
  year    = {2024},
  doi     = {10.1137/23M1598064}
}

@inproceedings{GZ23,
  author    = {Zeyu Guo and Zihan Zhang},
  title     = {Randomly Punctured {Reed--Solomon} Codes Achieve
               the List Decoding Capacity over Polynomial-Size Alphabets},
  booktitle = {2023 IEEE 64th Annual Symposium on Foundations
               of Computer Science (FOCS)},
  pages     = {164--176},
  year      = {2023},
  publisher = {IEEE},
  doi       = {10.1109/FOCS57990.2023.00019}
}

@inproceedings{AGL24,
  author    = {Omar Alrabiah and Venkatesan Guruswami and Ray Li},
  title     = {Randomly Punctured {Reed--Solomon} Codes Achieve
               List-Decoding Capacity over Linear-Sized Fields},
  booktitle = {Proceedings of the 56th Annual ACM Symposium
               on Theory of Computing},
  pages     = {1458--1469},
  year      = {2024},
  publisher = {ACM},
  doi       = {10.1145/3618260.3649634}
}

@article{BGM22,
  author  = {Joshua Brakensiek and Sivakanth Gopi and Visu Makam},
  title   = {Lower Bounds for Maximally Recoverable Tensor Codes
             and Higher Order {MDS} Codes},
  journal = {IEEE Transactions on Information Theory},
  volume  = {68},
  number  = {11},
  pages   = {7125--7140},
  year    = {2022},
  doi     = {10.1109/TIT.2022.3187366}
}

@ARTICLE{BDG23,
  author={Brakensiek, Joshua and Dhar, Manik and Gopi, Sivakanth},
  journal={IEEE Transactions on Information Theory}, 
  title={Improved Field Size Bounds for Higher Order MDS Codes}, 
  year={2024},
  volume={70},
  number={10},
  pages={6950-6960},
  doi={10.1109/TIT.2024.3449030}}

@techreport{AGG26,
  author      = {Omar Alrabiah and Rohan Goyal and
                 Venkatesan Guruswami},
  title       = {Explicit {Reed--Solomon} Subspace Designs:
                 Optimal List and Line Decoding},
  institution = {Electronic Colloquium on Computational Complexity},
  number      = {TR26-262},
  year        = {2026}
}

@misc{BCPZZ26,
  author       = {Joshua Brakensiek and Yeyuan Chen and Aaron Putterman
                  and Zihan Zhang and Kai Zhe Zheng},
  title        = {Algorithmic List Decoding of {Reed--Solomon} Codes
                  up to Capacity},
  howpublished = {arXiv preprint arXiv:2609.08005},
  year         = {2026},
  doi          = {10.48550/arXiv.2609.08005},
  url          = {https://arxiv.org/abs/2609.08005}
}

@inproceedings{DSY14,
  author    = {Son Hoang Dau and Wentu Song and Chau Yuen},
  title     = {On the Existence of {MDS} Codes over Small Fields
               with Constrained Generator Matrices},
  booktitle = {2014 IEEE International Symposium on Information Theory
               (ISIT)},
  pages     = {1787--1791},
  publisher = {IEEE},
  year      = {2014},
  doi       = {10.1109/ISIT.2014.6875141}
}

@inproceedings{Lov18,
  author    = {Shachar Lovett},
  title     = {{MDS} Matrices over Small Fields:
               A Proof of the {GM-MDS} Conjecture},
  booktitle = {2018 IEEE 59th Annual Symposium on Foundations
               of Computer Science (FOCS)},
  pages     = {194--199},
  publisher = {IEEE},
  year      = {2018},
  doi       = {10.1109/FOCS.2018.00027}
}

@article{YH19,
  author  = {Hikmet Yildiz and Babak Hassibi},
  title   = {Optimum Linear Codes with Support-Constrained
             Generator Matrices over Small Fields},
  journal = {IEEE Transactions on Information Theory},
  volume  = {65},
  number  = {12},
  pages   = {7868--7875},
  year    = {2019},
  doi     = {10.1109/TIT.2019.2932663}
}

@inproceedings{GHKSSY17,
  author    = {Parikshit Gopalan and Guangda Hu and Swastik Kopparty
               and Shubhangi Saraf and Carol Wang and Sergey Yekhanin},
  title     = {Maximally Recoverable Codes for Grid-like Topologies},
  booktitle = {Proceedings of the Twenty-Eighth Annual ACM-SIAM
               Symposium on Discrete Algorithms},
  pages     = {2092--2108},
  publisher = {SIAM},
  year      = {2017},
  doi       = {10.1137/1.9781611974782.136}
}

\end{document}